\documentclass[11pt]{article}
\usepackage[normalem]{ulem}
\usepackage{hyperref}
\usepackage{amsmath}
\usepackage{braket}
\usepackage{amsthm}
\usepackage{amssymb}
\usepackage{url}
\usepackage{algorithm}
\usepackage{algpseudocode}
\usepackage{mathtools}
\usepackage[title]{appendix}
\usepackage[all]{xy}
\usepackage{geometry}
\usepackage{physics}
\usepackage{cancel}
\usepackage{mdframed}
\usepackage[most]{tcolorbox}
\usepackage{enumitem}
\usepackage{xcolor}
\usepackage[dvipsnames]{xcolor}
\usepackage{comment}
\usepackage{array}
\usepackage{stackengine,scalerel}
\usepackage{dsfont}

\usepackage{complexity}

\usepackage{xcolor}
\definecolor{darkgreen}{RGB}{0, 100, 0}

\newtheorem{theorem}{Theorem}[section]

\newtheorem{lemma}[theorem]{Lemma}
\newtheorem{proposition}[theorem]{Proposition}

\newtheorem*{conjecture*}{\textbf{Conjecture}}
\newtheorem{corollary}[theorem]{Corollary}

\newtheorem{claim}{Claim}[theorem]

\theoremstyle{definition}
\newtheorem{definition}{Definition} 

\newtheorem{notation}{\textbf{Notation}}
\newtheorem{remark}[theorem]{Remark}
\newtheorem*{remark*}{\textbf{Remark}}

\newcommand{\Good}{\mathsf{Good}}
\newcommand{\Acc}{\mathsf{Acc}}
\newcommand{\Accept}{\mathsf{Accept}}
\newcommand{\Ver}{\mathsf{Ver}}
\newcommand{\SWAP}{\mathsf{SWAP}}

\newcommand{\iter}{\mathsf{iter}}

\newcommand{\rnd}{\mathsf{rnd}}
\newcommand{\negl}{\mathsf{negl}}

\newcommand{\QHE}{\mathsf{QHE}}
\newcommand{\HE}{\mathsf{HE}}
\newcommand{\Gen}{\mathsf{Gen}}
\newcommand{\Enc}{\mathsf{Enc}}
\newcommand{\Dec}{\mathsf{Dec}}
\newcommand{\Eval}{\mathsf{Eval}}

\newcommand{\Query}{\mathsf{Query}}

\newcommand{\ct}{\mathsf{ct}}
\newcommand{\pk}{\mathsf{pk}}
\newcommand{\sk}{\mathsf{sk}}
\newcommand{\eps}{\varepsilon}

\newcommand{\TD}{\mathsf{TD}}
\newcommand{\secp}{\lambda}
\newcommand{\QPT}{\mathsf{QPT}}

\newcommand{\brho}{\boldsymbol{\rho}}
\newcommand{\bsigma}{\boldsymbol{\sigma}}
\newcommand{\PPT}{\mathsf{PPT}}
\newcommand{\im}{\mathsf{im}}
\newcommand{\id}{\boldsymbol{I}}
\newcommand{\zero}{\boldsymbol{0}}

\newcommand{\Ber}{\mathsf{Ber}}

\makeatletter
\newcommand{\customlabel}[2]{%
   \protected@write \@auxout {}{\string \newlabel {#1}{{#2}{\thepage}{#2}{#1}{}} }%
   \hypertarget{#1}{#2}
}
\makeatother

\definecolor{classicrose}{rgb}{0.98, 0.8, 0.91}
\newcommand{\algo}[3]{
    \stepcounter{figure}
    \vspace{0.15cm}
    { \small
    \begin{tcolorbox}[breakable, enhanced, colback=classicrose!20]
    \begin{center}
    {\bf \underline{Algorithm~\customlabel{alg:#2}{\thefigure}: #1}}
    \end{center}
    
    #3
    \end{tcolorbox}
    }
}

\newcommand{\calA}{\mathcal A}
\newcommand{\calB}{\mathcal B}
\newcommand{\calC}{\mathcal C}
\newcommand{\calD}{\mathcal D}

\newcommand{\calH}{\mathcal H}

\newcommand{\calO}{\mathcal O}

\newcommand{\calS}{\mathcal S}
\newcommand{\calT}{\mathcal T}

\newcommand{\sfA}{\mathsf A}
\newcommand{\sfB}{\mathsf B}

\newcommand{\sfE}{\mathsf E}

\newcommand{\sfI}{\mathsf I}

\newcommand{\sfO}{\mathsf O}

\newcommand{\sfR}{\mathsf R}

\newcommand{\sfZ}{\mathsf Z}

\newcommand{\bbR}{\mathbb R}

\newcommand{\bbE}{\mathbb E}

\newcommand{\bbN}{\mathbb N}

\title{Tight Post-Quantum Parallel Repetition for Private-Coin Arguments}
\date{}
\author{Zvika Brakerski \thanks{Weizmann Institute of Science \texttt{zvika.brakerski@weizmann.ac.il}} \and Andrew Huang \thanks{University of California, Berkeley \texttt{a\_huang@berkeley.edu}} \and Yael Tauman Kalai \thanks{Massachusetts Institute of Technology \texttt{tauman@mit.edu}} \and Nicholas Spooner \thanks{Cornell University \texttt{nspooner@cornell.edu}}}
\begin{document}
\maketitle
\begin{abstract}
We show that assuming the existence of homomorphic encryption, parallel repetition of all interactive arguments (after being run under homomorphic encryption) reduces the soundness error at a tight rate even in the \emph{post-quantum} setting. Moreover, we generalize this result to hold for threshold verifiers, where the parallel repeated verifier accepts if and only if at least $t$ of the executions are accepted (for some threshold~$t$). Prior to this work, these results were known only when the cheating prover was assumed to be classical, and it was not known how to achieve tight bounds.

As a corollary, we construct the first constant-round succinct argument for $\QMA$ with negligible completeness and soundness errors assuming only the existence of quantum homomorphic encryption.
\end{abstract}
\newpage
\enlargethispage{1cm}
\tableofcontents
\pagenumbering{roman}
\newpage
\pagenumbering{arabic}

\section{Introduction}\label{sec:intro}
Parallel repetition is one of the most basic tools for amplifying the soundness of a proof system. The basic idea is to start from a protocol with noticeable soundness error and reduce this error by executing many independent copies of the protocol, with the verifier accepting only if sufficiently many copies accept. This repetition can be performed sequentially, which increases the number of rounds, or in parallel, which preserves the round complexity. This distinction is particularly important in cryptographic applications, where the number of rounds may be a central efficiency parameter.

In contrast to sequential repetition, whose soundness analysis is straightforward, the analysis of parallel repetition (even for interactive proofs) can be quite subtle. Although the verifier uses independent randomness for the different executions, a cheating prover may correlate its strategy across all executions, and hence the acceptance events need not be independent. A large body of work studies this problem for information-theoretic proof systems and games \cite{Raz95,Hol07}. In this work, we focus on  computationally sound interactive \emph{arguments}, where the problem is further complicated by the requirement that the reduction establishing soundness must itself be efficient.

For {\em public-coin} interactive arguments, where each message sent by the verifier is uniformly random, parallel repetition is by now well understood. H{\aa}stad, Pass, Pietrzak, and Wikstr{\"o}m \cite{HPPW10} proved that parallel repetition reduces the soundness error at an exponential rate against classical polynomial-time provers, even for protocols with a polynomial number of rounds. More recently, Huang and Kalai \cite{HK25b} established an analogous theorem against quantum polynomial-time provers.\footnote{Their result extends to the threshold setting, where the verifier accepts if the number of accepted executions exceeds some threshold.} The public-coin property is crucial in these proofs: Specifically, the proof uses the fact that one can efficiently simulate the verifier's messages given a transcript prefix.

For private-coin arguments, the situation is fundamentally more complicated. Indeed, Bellare, Impagliazzo, and Naor \cite{BIN97} showed that parallel repetition does not reduce the soundness error of general interactive arguments, even if the adversary is classical. Specifically, assuming trapdoor permutations, for every polynomial $k = \poly(\secp)$, they constructed a four-message argument whose soundness error remains essentially unchanged under $k$-fold parallel repetition. Thus, unlike in the public-coin setting, no general theorem can assert that simply repeating an arbitrary private-coin argument in parallel amplifies its soundness.  In a nutshell, the reason is that private-coin protocols can allow the parallel executions to help one another: a cheating prover may correlate the copies and feed information received from the verifier in one execution into another, causing the joint execution to succeed with much higher probability than the single-copy soundness would suggest.

Two generic approaches were developed to circumvent this negative result in the classical setting. Both approaches first modify the underlying protocol and then apply the parallel repetition to the modified protocol. The first is the random-termination approach due to Haitner\cite{Hai09}, which modifies the verifier so that it terminates at a randomly chosen point in the interaction. This makes the verifier ``partially simulatable'' and enables a parallel-repetition reduction. The second approach, due to Chung and Liu \cite{CL10}, uses fully homomorphic encryption. In the latter transformation, the verifier encrypts its messages, and the honest prover computes its responses homomorphically. They showed that the resulting protocol (and more generally any ``computationally simulatable'' protocol) admits parallel repetition against classical polynomial-time provers.

\subsection{Our Results} In this work, we study parallel repetition for private-coin interactive arguments in the setting where the cheating prover may be quantum. We follow the homomorphic-encryption approach of Chung and Liu \cite{CL10}. Given a private-coin interactive argument $(P,V)$, we apply their transformation to obtain an argument $(\widehat P,\widehat V)$ with the same number of rounds and, up to negligible additive error, the same completeness and single-copy soundness. When the honest prover is classical, the transformation uses a classical homomorphic encryption ($\HE$) scheme secure against quantum adversaries. When the honest prover is quantum, it uses a quantum $\HE$ scheme which is ``classical friendly'' in the sense that if we restrict this encryption scheme to classical messages, then the encryption and the decryption functions are classical.

Our main result establishes tight bounds in both the direct-product setting (where the verifier accepts only if all subverifiers accept) and the threshold parallel-repetition setting (where the verifier accepts if some minimum fraction of subverifiers accept). We do this by introducing a novel reduction strategy, which has the benefit of \emph{improving} upon the soundness established even in the classical setting, before quantizing the analyses used in the previous literature (specifically, the correlation reduction approach of \cite{CHS05,HS11}).

\begin{theorem}[Informal]\label{thm:tight_repetition}
Let $(P, V)$ be a private-coin interactive argument with a polynomial number of rounds and soundness error at most $\eps+\negl(\lambda)$ against quantum polynomial-time provers. Under the corresponding post-quantum $\HE$ assumption, let $(\widehat{P}, \widehat{V})$ be the corresponding encrypted protocol. For every polynomially bounded $k$, the soundness error of the all-accepting $k$-fold parallel repetition of $\widehat{V}$ is at most
    \[ \eps^k+\negl(\lambda). \] 
More generally, let $\widehat{V}^{k, t}$ execute $k$ copies of $\widehat{V}$ and accept if at least $t$ copies accept. Then the soundness error of the threshold repetition is at most
    \[ \Pr_{X_1, \ldots, X_k \sim \Ber(\eps)}\left[\sum_i X_i \geq t\right] + \negl(\lambda). \]
\end{theorem}

Both our direct-product and threshold results are information-theoretically tight (up to negligible factors, which are likely necessary \cite{DJMW12}) and \emph{improve} upon what is known even in the classical setting \cite{CL10,HS11,Chu11}. Notably, our threshold-verifier result gives a bound which is \emph{independent of round complexity}, provided that the protocol has polynomially many rounds. In contrast, previous results even in the classical setting based on the soft-decision approach of \cite{BIN97}, which remain state-of-the-art for certain parameter regimes, have bounds which degrade with the number of rounds (see \cite{HPPW10,Chu11} for details).

The main difficulty in extending prior works to quantum provers is that their classical reduction relies on rewinding. Roughly, the reduction first runs the repeated prover using an encryption of a dummy verifier message, identifies a favorable execution, and then rewinds the prover and replaces the dummy encryption with an encryption of the actual message received from an external verifier. Such rewinding is not available for an arbitrary quantum prover, whose state may moreover be entangled across all repetitions.

There is also a subtler cryptographic difficulty. Although encryptions of the dummy message and the real verifier message are computationally indistinguishable, the reduction must first condition the prover's state on a favorable acceptance event using the dummy ciphertext, and then continue from the resulting state using the real ciphertext. Ordinary semantic security does not by itself guarantee that this state remains suitable for the real-ciphertext continuation. Our proof overcomes these obstacles by developing a coherent substitute for the classical rewinding procedure, together with a ciphertext-replacement argument that remains valid under the relevant coherent conditioning. We defer a detailed description of these ideas to the technical overview in Section~\ref{sec:overview}.

As an immediate application, we can apply our parallel repetition theorem to existing constant-round succinct arguments for $\QMA$ and general compilation of $\MIP^{*}$s \cite{MNZ24, GKNV25, HK25a}. 
\begin{corollary}[Informal]
    Assuming the existence of $\QHE$ (which is implied by the post-quantum hardness of LWE), there exist constant-round succinct arguments for $\QMA$ with negligible completeness and soundness errors.
\end{corollary}
Previous constructions of constant-round succinct arguments for $\QMA$ with negligible soundness required the much stronger assumption of post-quantum indistinguishability obfuscation as well as LWE \cite{BKL+22}.

\section{Technical Overview}\label{sec:overview}
\newcommand{\base}{V}
\newcommand{\rep}{\heV^k}
\newcommand{\adv}{\mathsf{A}}
\newcommand{\red}{\mathsf{R}}
\newcommand{\kmi}{[k]\setminus\{i\}}
\newcommand{\hq}{\hat{q}}
\newcommand{\hbase}{\widehat{\base}}

\newcommand{\rR}{\mathsf{R}}
\newcommand{\rE}{\mathsf{E}}

\newcommand{\heP}{\widehat P}
\newcommand{\heV}{\widehat V}

\subsection{HE Transformation}
We begin by describing the $\HE$ transformation originally due to Chung and Liu \cite{CL10}. Given an interactive argument system $\langle P, V \rangle$ and an $\HE$ scheme $(\Gen,\Enc,\Dec,\Eval)$, we obtain a new argument system $\langle \heP, \heV \rangle$ as follows. In each round, the verifier $\heV$ simulates the original verifier $V$, encrypting its message under $\HE$ (using a fresh key pair in each round), and forwarding the ciphertext and public key to the prover.\footnote{For technical reasons, $\heV$ also sends the secret key for the previous round.} $\heP$ homomorphically evaluates $P$ on the received messages and return the (encrypted) responses to the verifier. On receiving a response, $\heV$ decrypts it and forwards the plaintext to $V$. Completeness and soundness of $(\heP,\heV)$ follows easily from the completeness and soundness of $(P,V)$ and correctness of $\HE$.

Chung and Liu show that the soundness of the $k$-fold parallel repetition of $\heV$ decreases exponentially in $k$. Their result applies to a broader class of protocols with ``simulatable'' verifiers, for which there is an efficient algorithm whose output, on input a partial transcript, is indistinguishable from the verifier's next message given that transcript. This holds for $\heV$ due to semantic security of $\HE$: the simulator can simply supply an encryption of zero. Unfortunately, it is unknown how to transfer this argument to the (post-)quantum setting. Instead we develop a new proof technique which is tailored to the $\HE$ transformation and consequently more amenable to quantization. As a bonus, our new technique also yields much tighter bounds than what was known in the classical setting.

\subsection{Classical analysis}
Our starting point is a new analysis in the classical setting. For the sake of simplicity, our focus will be on the direct-product setting, where the repeated verifier accepts only if all of its subverifiers accept. Our goal is to construct an efficient reduction $\red$ that takes as a resource an adversary $\adv$ with success probability $\gamma = 1/\poly(\lambda)$ against $\rep$, the parallel-repeated argument under $\HE$, and converts it into an adversary $\red^{\adv}$ against $V$ succeeding with probability close to $\gamma^{1/k}$.\footnote{Note that, unlike prior analyses via simulatability, the target of our reduction is $V$ rather than $\heV$.}

The reduction $\red$ interfaces with a \emph{external verifier} $\base$ and has black-box access to $\adv$. The basic strategy is to sample an index $i \in [k]$ ahead of time and embed $\base$ in the $i$-th instance of an execution of $\rep$. On all other instances $\kmi$, the reduction will simulate $\heV$ internally with an appropriately-chosen random tape; we refer to those instances as the \emph{internal verifiers}.

The general structure of the reduction is as follows: in round $j$, the reduction obtains a query $q^j$ from the external verifier and encrypts it under some randomness to obtain $\hq^j$, computes the encrypted queries $\hq_{-i}$ of the internal verifiers, and forwards the vector $\hq$, obtained by inserting $\hq^j$ into the $i$-th position of $\hq_{-i}$, to $\adv$. On receipt of encrypted responses $\hat a$, $\red$ decrypts $\hat a$ and checks if $a$ corresponds to a set of good responses (i.e. the prover's success probability conditioned on these responses remains relatively high) before forwarding $a_i$ to the external verifier.

There are two major choices we must make in the design of our reduction. First, how do we decide what index $i$ to embed the external verifier in? Second, how should we choose the randomness for the internal verifiers $\heV$ (including their encryption randomness) and the randomness used to encrypt the external verifier's message? 

Since $\langle P, V \rangle$ is a private-coin protocol, we do not have access to the external verifier's randomness, and thus we cannot directly determine whether it will accept. To find a proxy for this event, a natural first attempt would simply be to check whether the internal verifiers accept. After all, as the cheating prover must by assumption convince all $k$ subverifiers with probability $\gamma$, succeeding in executions $\kmi$ might be a signal that execution $i$ is also good. However, as Chung \cite{Chu11} points out, this approach does not work out if the prover's ``success pattern'' in coordinate $i$ is correlated with its behavior in other coordinates. Indeed, an adversary that convinces all $k$ verifiers with probability $\gamma$ and otherwise chooses $j \gets [k]$ uniformly and convinces all but the $j$-th verifier would cause the above strategy to fail: conditioning on the all but the $i$-th verifier accepting will bias the distribution towards the latter case for small $\gamma$.

Fortunately, as Canetti, Halevi, and Steiner \cite{CHS05} observe, such bad correlations allow us to directly convert $\adv$ into a cheating prover which interacts with $\heV^{k-1}$ that has success probability at least $\gamma^{(k-1)/k}$. In our previous example, if we simply simulate the first coordinate internally and embed the $k-1$ coordinates of $\heV^{k-1}$ in coordinates $2$ through $k$, it is straightforward to see that our adversary convinces $\heV^{k-1}$ with success probability $\gamma + (1-\gamma)/k \gg \gamma^{(k-1)/k}$. 

Thus, following the correlation reduction strategy of \cite{CHS05,CL10,HS11}, we run a pre-processing step to ensure that no coordinates have bad correlations. In particular, we check that the first coordinate does not have bad correlations with all other coordinates before proceeding (and reducing the number of repetitions via simulation if such correlations do exist), thereby fixing $i = 1$. With some work, one can show that once (most of) the bad correlations have been removed, our naive rule actually serves as a good proxy for when $a_i$ is a good response.

Specifically, let $\heV^k_i(q^1,\ldots,q^m;r_{-i},\rho_1,\ldots,\rho_m)$ be the $k$-fold repetition of $\heV$ where the verifiers in $\kmi$ use randomness $r_{-i}$, the $i$-th verifier sends an encryption of $q^j$ in round $j$ using randomness $\rho_j$, and the final decision is made by checking all verdicts besides the $i$-th verifier. At the beginning of the protocol, we repeatedly sample randomness $(r_{-1}, \rho_1, \ldots, \rho_m)$ for our verifier until it accepts when interacting with $\adv$ and store $r_{-1}$ which will remain fixed through the remainder of the protocol. Since we do not have the external verifier's actual randomness at this stage, we have $\heV^k_1$ use $\rho_j$ to encrypt $0$.

In the $j$-th round of the interaction, on receiving message $q^j$ from the external verifier, $\red$ will choose the $j$-th round encryption randomness $\rho_j$ by repeatedly resampling $(\rho_j,\ldots,\rho_m)$ until $\heV^k_1(q^1,\ldots,q^j,0,\ldots,0;r_{-1},\rho_1,\ldots,\rho_m)$ accepts when interacting with $\adv$.

We argue that $\red^\adv$ convinces the external verifier with good probability. Consider a thought experiment where we do have access to the random tape of the external verifier, and where throughout we replace the $0$ inputs to $\heV^k_1$ with the correct $q^j$. Since our acceptance heuristic/rule does not depend on the first coordinate, semantic security of $\HE$ implies that the resulting distribution of $(r_{-1}, \rho_1, \ldots, \rho_m)$ is \emph{statistically} close between these two cases. Observe now that in the thought experiment, each round simply corresponds to resampling the suffix $(\rho_j, \ldots, \rho_m)$ under the same conditioning as before. Thus we obtain the same randomness distribution if we simply fix $(r_{-1}, \rho_1, \ldots, \rho_m)$ as sampled at the beginning of the protocol. At this point standard analysis implies that the external verifier accepts with the desired probability.

\subsection{Quantum analysis}
Let us now move to the quantum setting, where both the honest prover and the adversary are potentially quantum algorithms (though the verifier remains classical). We will closely follow the classical proof strategy outlined above.

Recall that before the interaction begins, we sample randomness $r_{-1}$ for the internal verifiers using a rejection sampling procedure. We follow essentially the same strategy in the quantum setting: there is no need to rewind here because we can assume we are given advice consisting of many copies of $\adv$'s initial state.\footnote{This is necessary for a reduction that increases the success probability of the adversary; see \cite{BQSY24} for details.} The key difference is that rather than \emph{sampling} the encryption randomnesses $\rho_1,\ldots,\rho_m$ for the first coordinate (corresponding to the external verifier), we keep them in superposition, and coherently measure whether $\heV^k_1(0,\ldots,0)$ accepts (although we still sample $r_{-1}$ classically as before). After polynomially many attempts, we will successfully ``post-select'' on $\heV^k_1(0, \ldots, 0)$ accepting.

The more challenging part to handle is the interaction: once we have started to interact with the external verifier we are no longer able to ``throw away'' the adversary's state and try again. Notice that, via the same argument as in the classical setting, if we had performed the above post-selection based on $\heV^k_1(q^1,\ldots,q^m)$ where $q^1, \ldots, q^m$ are the external verifier's \emph{real} messages, the reduction would succeed with good probability. Thus it suffices to produce a state for the adversary that is statistically close to this ``ideal'' post-selection.

As in the classical reduction, the idea will be to gradually ``swap in'' the $q^j$ as they arrive, in a way that the adversary cannot detect, relying on the semantic security of $\HE$. Denote by $\ket*{\Enc(m)}_{\rR\rE} \propto \sum_{r} \ket{r}_{\rR} \ket*{\Enc(m;r)}_{\rE}$ the uniform superposition of encryptions of $m$. Notice that because $\rR$ is outside of its view, $\adv$ cannot distinguish $\ket*{\Enc(0)}$ from $\ket*{\Enc(q^j)}$. The difficulty is that the initial post-selection causes $\adv$ to become entangled with $\rE$. This entanglement must be preserved when we switch the ciphertext.

To achieve this we borrow and extend a technique from Lombardi, Ma and Spooner \cite{LMS22}, there used in the context of zero knowledge simulation. The entangling operation is a binary projective measurement (of whether $\heV^k_1$ accepts); i.e., a projection $\Pi_A$ on to the subspace $A$ of ``accepting'' states. We also consider the subspace $B$ of states where the $\rR\rE$ registers hold $\ket*{\Enc(0)}$ (in particular, they are unentangled with the adversary); the initial state of the reduction belongs to this subspace, and the corresponding binary projective measurement $\Pi_B = \ketbra*{\Enc(0)} \otimes I$ can be implemented efficiently.

A common principle in ``quantum rewinding'' reductions is that a state in a subspace $B$ that is perturbed by a projective measurement with respect to subspace $A$ can be efficiently rotated back to $B$ given oracles for the projections $\Pi_A,\Pi_B$. The algorithm that achieves this is the Quantum Singular Value Transformation (QSVT) \cite{GSLW19}. Let $\Pi_A,\Pi_B$ be projections on to $A,B$ respectively; we can write $\Pi_B \Pi_A = \sum_j \sigma_j \ketbra{v_j}{w_j}$ where $\sigma_j = \braket{v_j}{w_j} \in \mathbb{R}$ are the singular values of $\Pi_B\Pi_A$. The QSVT approximately implements the (partial) unitary transformation $\sum_{j, \sigma_j > \tau} \ketbra{v_j}{w_j}$ in time $\poly(1/\tau)$. That is, the QSVT will rotate into $B$ components of a state in $A$ with sufficiently high singular value. Notice that for a state $\ket{\psi} = \sum_j \alpha_j \ket{v_j} \in B$, the amplitudes of the postselected state $\ket{\psi'} \propto \Pi_A \ket{\psi} = \sum_j \alpha_j \sigma_j \ket{w_j}$ are skewed away from components where $\sigma_j$ is small. By a standard Markov argument we are able to rotate $\ket{\psi'}$ back into $B$ up to error $\varepsilon$ in time $\poly(1/\varepsilon)$.

Following \cite{LMS22}, our ``swapping'' operation works as follows. After the initial postselection step, we use the QSVT to rotate the state back to $B$, so that the $\rR\rE$ registers hold (approximately) $\ket*{\Enc(0)}$. Next we switch out the state on $\rR\rE$ with $\ket*{\Enc(q^j)}$. We then apply the \emph{inverse} of the QSVT, with the change that $\Pi_B$ is replaced with the projection $\Pi_{B'} = \ketbra*{\Enc(q^j)} \otimes I$ on to the subspace $B'$ of states where the $\rR\rE$ registers hold $\ket*{\Enc(q^j)}$.

Why does this work? Suppose for a moment that $\ket*{\Enc(0)}$ and $\ket*{\Enc(q^j)}$ have the same reduced density matrix on $\rE$. Then there is an isometry $V$ acting only on $\rR$ which maps $\ket*{\Enc(0)}$ to $\ket*{\Enc(q^j)}$. Up to small errors, the state-switching operation in the above procedure is equivalent to an application of $V$. We can then imagine commuting $V$ back through the QSVT: it commutes with $\Pi_A$ and conjugates $\Pi_B$ to $\Pi_{B'}$. Now the QSVT and inverse QSVT are both with respect to $\Pi_{B'}$, and so they cancel out. Finally, $V$ applied to the initial state $\ket*{\Enc(0)}_{\rR\rE} \ket{\phi}$, where $\ket{\phi}$ is the initial state of $\adv$, maps it to $\ket*{\Enc(q^j)} \ket{\phi}$, and so the resulting state is $\ket{\phi'} \propto \Pi_A \ket*{\Enc(q^j)} \ket{\phi}$ as desired.

Of course, there is no such isometry: $\ket*{\Enc(0)}$ and $\ket*{\Enc(q^j)}$ have disjoint support on $\rE$. However, their reduced densities are \emph{computationally} indistinguishable by semantic security of $\HE$. Note that it is not immediately clear how semantic security should apply here, since the QSVT has access to the projections $\Pi_B,\Pi_B'$ on to these states. Nonetheless, \cite{LMS22} shows how to carry out the above argument in the computational setting, showing that the final state is \emph{computationally indistinguishable} from $\ket{\phi'}$.

There is a further subtlety in our setting, however: computational indistinguishability will not suffice here, because to determine whether $\red^\adv$ convinces the external verifier requires the secret key of the encryption. Thus (similarly to the classical setting) we need to argue that the final state is \emph{statistically} indistinguishable from $\ket{\phi'}$. We achieve this by ``lifting'' statistical to computational distinguishability via a SWAP test; see Section \ref{sec:oracle_distinguishers} for details.

Finally, by using a pre-processing procedure due to Holenstein and Schoenebeck \cite{HS11} and slightly modifying our acceptance heuristic, we can extend our reduction to work for threshold verifiers in a fairly straightforward manner.

\section{Preliminaries}\label{sec:preliminaries}
We use $\PPT$ as a shorthand for probabilistic polynomial time and $\QPT$  as a shorthand for quantum polynomial time. When referring to non-uniform quantum circuits, we always refer to quantum circuits with \emph{quantum} (and not classical) advice unless otherwise specified. 

\begin{notation}
    For any two quantum states $\brho_1$ and $\brho_2$, we use the notations
    \[
    \brho_1\stackrel{\eps}\equiv\brho_2 ~~\mbox{ and }~~\TD(\brho_1,\brho_2)\leq\eps
    \]
    to denote that the trace distance between $\brho_1$ and $\brho_2$ is at most $\eps$. In particular, this means that for any (possibly unbounded) quantum adversary $\calA$,
    \[
    |\Pr[\calA(\brho_1)=1]-\Pr[\calA(\brho_2)=1]| \leq \eps.
    \]
    For any two families of quantum states $\brho_1=\{\brho_{1,\secp}\}_{\secp\in\mathbb{N}}$ and $\brho_2=\{\brho_{2,\secp}\}_{\secp\in\mathbb{N}}$
    we use the notation
        \[ \brho_1\stackrel{\eps}\approx\brho_2 \]
    to denote that $\brho_1$ and $\brho_2$ are $\eps$-computationally indistinguishable, which means that for any $\QPT$ adversary $\calA$ there exists a negligible function $\negl$, such that for every $\secp \in \bbN$,
        \[ |\Pr[{\cal A}(\brho_{1,\secp})=1]-\Pr[{\cal A}(\brho_{2,\secp})=1]|\leq \eps(\secp)+\negl(\secp). \]
\end{notation}

\subsection{Information Theory} \label{sec:info}
\begin{notation}\label{notation:bin}
    For any $\eps \in [0, 1]$ and integers $t$ and $k$, define the threshold binomial function by 
        \[ B(k, t, \eps) := \sum_{s = t}^k \binom{k}{s} \eps^s (1-\eps)^{k-s}. \]
    Note that $B(k, k, \eps) = \eps^k$, and that $B(k, t, \eps)$ is exactly the probability that the sum of $k$ i.i.d. Bernoulli random variables with parameter $\eps$ is at least $t$.
\end{notation}

We will use Hoeffding's inequality in this work:
\begin{proposition}[Hoeffding's inequality]\label{proposition:hoeffding}
Let $X_1, \ldots, X_n$ be independent bounded random variables with $X_i \in [a, b]$ for all $i$, where $-\infty < a \leq b < \infty$. Then 
    \[ \Pr[\left|\frac{1}{n} \sum_{i=1}^n (X_i-\bbE[X_i])\right| \geq t] \leq 2e^{-2nt^2/(b-a)^2}. \]
\end{proposition}

\subsection{Interactive Protocols with Computational Soundness}\label{prelim:IA}

In what follows, we focus on interactive protocols between two parties, where one party, denoted by $V$ (for verifier), outputs a single bit, indicating acceptance or rejection. The other party is denoted by $P$ (for prover) and tries to convince the verifier to accept. We define the notion of computational soundness for such protocols. For the sake of simplicity, we focus on the case where the parties receive as input the security parameter $1^\secp$, and omit all additional inputs. We denote by
    \[ \langle P, V \rangle(1^\secp) \]
the output bit of $V$ after interacting with $P$.

Given an interactive protocol $(P, V)$, we can also define the $k$-fold parallel repetition with threshold $t$ (denoted by $(P^{\otimes k}, V^{k, t})$) as follows: the verifier $V^{k, t}$ executes $k$ independent executions/subverifiers $V_i$ (for $i \in [k]$) in parallel and accepts if and only if at least $t$ out of the $k$ subverifiers accept. The honest prover $P^{\otimes k}$ simply executes $k$ copies of the honest prover strategy $P$ in parallel. For a general prover $\calA$ (not necessarily honest) we denote by
    \[ \langle\calA, V^{k, t}\rangle(1^\secp) \]
the output bit of $V^{k, t}$ after interacting with $\calA$.

\begin{remark}
      We assume without loss of generality that the first message is sent by the verifier, as otherwise the verifier can always begin by sending a dummy message which is ignored in the rest of the protocol. We also assume that the last message is sent by the prover. We call an interactive protocol an $m$-round protocol if it consists of $m$ rounds of back and forth messages from the verifier to the prover (as opposed to the usual $2m$-round definition). 
\end{remark}

\begin{notation}
    Denote the transcript of an $m$-round interactive protocol $(P, V)$ as all messages sent between the prover and verifier, labeled $\tau = (q_1, a_1, \dots, q_m, a_m)$, where $q_j$ is the query sent by~$V$ in the $j$'th round, and $a_j$ is the response sent by $P$ in the $j$'th round. We let $(\tau, r) \gets (P, V)$ denote the random variable corresponding to the transcript and internal randomness, respectively, of a random interaction of a prover $P$ with verifier $V$.
\end{notation}  

\begin{notation}
    Given a transcript $\tau$ and verifier randomness $r$, we denote the verdict function of $V$ by $V(\tau, r) \in \{0, 1\}$ to be the (deterministic) output bit of $V$. For a threshold verifier $V^{k, t}$, we let the randomness of the $i$'th subverifier be denoted by $r_i$ and the transcript of the $i$'th execution be denoted by $\tau_i$.
\end{notation}

\subsection{Homomorphic Encryption}\label{sec:HE}
In this section, we define the two types of encryption schemes we consider in this work. 
\begin{definition}[Classical Homomorphic Encryption]
    Let $\secp$ be the security parameter. A homomorphic encryption scheme $\HE$ consists of four $\PPT$ algorithms 
        \[ \HE = (\HE.\Gen, \HE.\Enc, \HE.\Dec, \HE.\Eval) \]
    with the following syntax:
    \begin{itemize}
        \item \textbf{Key Generation.} The probabilistic algorithm $(\pk, \sk) \gets \HE.\Gen(1^{\secp})$ outputs a public key $\pk$ and a secret key $\sk$.
        \item \textbf{Encryption.} The probabilistic algorithm $\ct \gets \HE.\Enc(\pk, m)$ takes as input a public key $\pk$ and message $m \in \{0, 1\}^{\poly(\lambda)}$ and outputs a ciphertext $\ct$.
        \item \textbf{Decryption.} The deterministic algorithm $m \gets \HE.\Dec(\sk, \ct)$ takes as input a secret key $\sk$ and ciphertext $\ct$ and outputs a message $m$.
        \item \textbf{Evaluation.} The (deterministic or probabilistic) algorithm $\ct' \gets \HE.\Eval(\pk, C, \ct)$ takes as input a public key $\pk$, circuit $C: \{0, 1\}^{\poly(\lambda)} \to \{0, 1\}^{\poly(\lambda)}$, and ciphertext $\ct$, and outputs a ciphertext $\ct'$.  
    \end{itemize}
    In addition, $\HE$ should satisfy the following properties:
    \begin{enumerate}
        \item \textbf{Correctness.} $\HE$ is \emph{fully homomorphic} if for any efficiently computable circuit $C: \{0, 1\}^{\poly(\lambda)} \to \{0, 1\}^{\poly(\lambda)}$ and any input $x$,
            \[ \Pr[\HE.\Dec(\sk, \ct') \neq C(x)] = \negl(\lambda), \]
        where the probability is taken over $(\pk, \sk) \gets \HE.\Gen(1^{\secp})$, $\ct \gets \HE.\Enc(\pk, x)$, and $\ct' \gets \HE.\Eval(\pk, C, \ct)$. $\HE$ is \emph{leveled fully homomorphic} if it takes $1^L$ as additional input in key generation, and can only evaluate depth-$L$ Boolean circuits.
        \item \textbf{(Non-Uniform Post-Quantum) Semantic Security.} $\HE$ is secure if for any $\poly(\lambda)$-size non-uniform quantum adversary $\calA$ and any pair of messages $m_0, m_1$ such that $|m_0| = |m_1| = \poly(\secp)$ it holds that
            \[ |\Pr[\calA(\pk, \HE.\Enc(\pk, m_0)) = 1]-\Pr[\calA(\pk, \HE.\Enc(\pk, m_1)) = 1]| = \negl(\secp), \]
        where $(\pk, \sk) \gets \HE.\Gen(1^{\secp})$.\footnote{Technically speaking, semantic security often allows $m_0, m_1$ to depend on $\pk$, but we will only need security against non-adaptively chosen plaintexts.}\footnote{To be precise, we will in general need semantic security against quantum adversaries with quantum advice/non-uniformity.}
        \item \textbf{Compactness.} $\HE$ is compact if its decryption algorithm $\HE.\Dec$ is independent of size of the evaluated circuit.\footnote{On the other hand, the decryption algorithm must be allowed to run in time polynomial in the output length of the circuit.} In the case that $\HE$ is only leveled fully homomorphic, $\HE.\Dec$ can also run in time polynomial in the level upper bound $L$.
    \end{enumerate}
\end{definition}

\begin{definition}[Quantum Homomorphic Encryption (QHE), \cite{GV24}]
    A quantum homomorphic encryption scheme consists of four $\QPT$ algorithms  
        \[ (\QHE.\Gen, \QHE.\Enc, \QHE.\Eval, \QHE.\Dec) \]
    with the following syntax:
        \begin{itemize}
            \item \textbf{Key Generation.} $(\pk, \sk) \gets \QHE.\Gen(1^{\secp})$ is a $\PPT$ algorithm that takes as input the security parameter $1^{\lambda}$ and outputs a public key $\pk$ and secret key $\sk$.
            \item \textbf{Encryption.} $\brho \gets \QHE.\Enc(\pk, \bsigma)$ is a $\QPT$ algorithm which takes as input a public key $\pk$ and a state $\bsigma$, and outputs a ciphertext $\brho$.
            \item \textbf{Evaluation.} $\bsigma \gets \QHE.\Eval(\pk, C, \brho)$ is a $\QPT$ algorithm which takes as input a public key $\pk$, a quantum circuit $C$, and a ciphertext state $\brho$, and outputs a state $\bsigma$.
            \item \textbf{Decryption.} $\brho \gets \QHE.\Dec(\sk, \bsigma)$ is a $\QPT$ algorithm which takes as input a secret key $\sk$ and a quantum ciphertext $\bsigma$, and outputs a quantum state $\brho$.
        \end{itemize}
    In addition, $\QHE$ should satisfy the following properties:
        \begin{enumerate}
            \item \textbf{Correctness.} $\QHE$ is \emph{fully homomorphic} if there exists a negligible function $\mu$ such that, for all $\secp \in \bbN$, all $\poly(\secp)$-size quantum circuits $C$, and all $(\sk, \pk) \gets \QHE.\Gen(1^{\secp})$,
                \[ \left\| \QHE.\Dec_{\sk} \circ \QHE.\Eval_{\pk,C} \circ \QHE.\Enc_{\pk} - \calC \right\|_{\diamond} \leq \mu(\secp), \]
            where $\calC$ is the quantum channel implemented by $C$.
            
            $\QHE$ is \emph{leveled fully homomorphic} if it takes $1^L$ as additional input in key generation, and can only evaluate quantum circuits of $T$-depth at most $L$.
            \item \textbf{(Non-Uniform) Semantic Security.} $\QHE$ is secure if for all $\poly(\secp)$-sized non-uniform quantum adversaries $\calA$ and any pair of messages $m_0, m_1$ such that $|m_0| = |m_1| = \poly(\secp)$,\footnote{One can ask for a stronger form of security, where $\calA$ is allowed to ask for encryptions of quantum messages, but we will only need this weaker form of security with classical ciphertexts.}
                \[ |\Pr[\calA(\pk, \QHE.\Enc(\pk, m_0)) = 1]-\Pr[\calA(\pk, \QHE.\Enc(\pk, m_1)) = 1]| = \negl(\secp), \]
            where $(\pk, \sk) \gets \QHE.\Gen(1^{\secp})$.
            \item \textbf{Compactness.} $\QHE$ is compact if its decryption algorithm $\QHE.\Dec$ is independent of size of the evaluated circuit. In the case that $\QHE$ is only leveled fully homomorphic, $\QHE.\Dec$ can also run in time polynomial in the level upper bound $L$.
            \item \textbf{Classically Friendly.} $\QHE$ is \emph{classically friendly} if for all classical messages $m$ and classical ciphertexts $c$, $\QHE.\Enc(\pk, m)$ and $\QHE.\Dec(\sk, c)$ are $\PPT$ algorithms, and moreover, $\QHE.\Dec(\sk, c)$ is deterministic.\footnote{We emphasize that this condition is relatively mild and is only needed for the correctness of our transformation. In particular, all known constructions of $\QHE$ \cite{Mah20,Bra18,GV24} satisfy this property, and protocols compiled using the compiler of \cite{KLVY23} assume the existence of $\QHE$ with this property.}
        \end{enumerate}
\end{definition}

\subsubsection{The Homomorphic-Encryption Transformation}
Let $(P, V)$ be a $m$-round interactive protocol. Without loss of generality, we can assume that $(P, V)$ has the following syntax:
\begin{enumerate}
    \item $V$ samples randomness $r_{\Query} \gets \{0, 1\}^*$.
    \item For $i \in [m]$:
    \begin{enumerate}
        \item $V$ computes the message $q_i = f_i(r_{\Query}, a_1, \ldots, a_{i-1})$ and sends $q_i$ to $P$.
        \item Upon receiving $q_i$, $P$ computes some (not necessarily deterministic) response $a_i$ and sends $a_i$ to $V$.
    \end{enumerate}
    \item $V$ computes some function $b = g(r_{\Query}, q_1, a_1, \ldots, q_m, a_m) \in \{0, 1\}$ and outputs $b$.
\end{enumerate}

\noindent We now define the $\QHE$-transformed protocol $(\widehat{P}, \widehat{V})$ as follows:
\begin{enumerate}
    \item $\widehat{V}$ samples $r_{\Query}$ as well as randomness $r_{\Gen}, (r_{\Enc, 1}, \ldots, r_{\Enc, m}), (r'_{\Enc, 2}, \ldots, r'_{\Enc, m})$ used for encryption. Using $r_{\Gen}$, $\widehat{V}$ computes encryption keys $\{(\pk_i, \sk_i)\}_{i \in [m]}$.
    \item For $i \in [m]$:
    \begin{enumerate}
        \item $\widehat{V}$ computes $a_{i-1} = \QHE.\Dec(\sk_{i-1}, \widehat{a}_{i-1})$ and its message $q_i = f_i(r_{\Query}, a_1, \ldots, a_{i-1})$. It then encrypts $q_i$ to produce $\widehat{q}_i = \QHE.\Enc(\pk_i, q_i; r_{\Enc, i})$. $\widehat{V}$ sends $\pk_i$ and $\widehat{q}_i$ to $P$. If $i \neq 1$, $\widehat{V}$ also computes $\widehat{\sk}_{i-1} \gets \QHE.\Enc(\pk_i, \sk_{i-1}; r'_{\Enc, i})$ and sends $\widehat{\sk}_{i-1}$ to $\widehat{P}$.
        \item Upon receiving $\pk_i$ and $\widehat{q}_i$ (as well as $\widehat{\sk}_{i-1}$ when $i > 1$), $\widehat{P}$ computes some (not necessarily deterministic) response $\widehat{a}_i$ and sends $\widehat{a}_i$ to $\widehat{V}$.
    \end{enumerate}
    \item $\widehat{V}$ decrypts $\widehat{a}_m$ to $a_m = \QHE.\Dec(\sk, \widehat{a}_m)$ before computing $b = g(r_{\Query}, q_1, a_1, \ldots, q_m, a_m)$ and outputting $b$.
\end{enumerate}
\begin{remark}
    Note that when the honest prover $P$ is classical, it suffices to use a classical homomorphic encryption scheme $\HE$ instead of $\QHE$. This is the case when we are considering just the post-quantum security of classical protocols rather than protocols where even honest provers must be quantum (as in \cite{Mah18,MNZ24,GKNV25,HK25b}).
\end{remark}

\begin{remark}
    The use of a fresh encryption key in every round is important for the security analysis. Looking ahead, when replacing the encrypted verifier message in round $i$ by an encryption of zero, we condition on the interaction up to round $i-1$. If the same key were reused throughout the protocol, this conditioning would also fix the key under which the next message is encrypted, and thus we would no longer be able to rely on semantic security. By using a fresh key $\pk_i$, the encryption key for round $i$ remains independent of the preceding interaction which allows us to rely on semantic security. The same fresh-key transformation was already used in the classical setting by Chung and Liu~\cite{CL10} for similar reasons.
\end{remark}

\begin{lemma}\label{lemma:qfhe_transform}
    $(P, V)$ and $(\widehat{P}, \widehat{V})$ have the same completeness and soundness, up to negligible (additive) error. Additionally, $(P^{\otimes k}, V^{k, t})$ and $(\widehat{P}^{\otimes k}, \widehat{V}^{k, t})$ have the same completeness for any $t, k = \poly(\secp)$ up to a negligible (additive) error. 
\end{lemma}
\begin{proof}
    To establish that $(\widehat{P}, \widehat{V})$ has completeness at least as good (up to negligible additive error) as $(P, V)$, we simply have $\widehat{P}$ implement the honest prover $P$'s strategy under the hood using $\QHE$ (or $\HE$). Specifically, in each round, $\widehat{P}$ uses $\widehat{\sk}_{i-1}$ to recompute the prior transcript under $\pk_i$ before applying the honest prover $P$'s strategy for the $i$th round with $\QHE$ (resp. $\HE$) with respect to $\pk_i$. By the correctness of $\QHE$ (resp. $\HE$), $\widehat{P}$ has the same success probability as $P$ up to negligible factors. Note that the classical friendliness of $\QHE$ also guarantees that $\widehat{V}$ remains classically efficient.

    Showing the other direction (i.e. that $(P, V)$ has completeness at least as good as $(\widehat{P}, \widehat{V})$) is even simpler: $P$ simply generates $\HE/\QHE$ keys for each round, encrypts every query $q_i$ it gets from the verifier $V$, runs $\widehat{P}$ on $\widehat{q}_i$ to get a response $\widehat{a}_i$, and sends the decryption $a_i$ of $\widehat{a}_i$ back to $V$. Completeness follows from the fact that a random interaction of $P$ with $V$ perfectly simulates an interaction of $\widehat{P}$ with $\widehat{V}$. Similarly, it is not hard to see that soundness follows from direct simulation arguments and the correctness of $\QHE$.
    
    Establishing completeness for the threshold parallel repetitions $(P^{\otimes k}, V^{k, t})$ and $(\widehat{P}^{\otimes k}, \widehat{V}^{k, t})$ follows from an essentially identical argument since the honest provers $P^{\otimes k}$ and $\widehat{P}^{\otimes k}$ and verifiers all act independently across executions.
\end{proof}
\begin{remark}
   Note that the same argument does \emph{not} imply that the soundness of $(P^{\otimes k}, V^{k, t})$ can be reduced to that of $(\widehat{P}^{\otimes k}, \widehat{V}^{k, t})$. This is since a cheating prover in the (unencrypted) parallel-repeated protocol can correlate its answers across executions. As the verifier's messages in $(\widehat{P}^{\otimes k}, \widehat{V}^{k, t})$ are encrypted under different keys per execution, it is not clear how to simulate such a correlated strategy anymore.
\end{remark}

The following lemma, adapted from \cite{HS11} (see also \cite{Chu11}), serves as a building block in the proof of the (classical) parallel-repetition theorem. As we show below, this lemma readily extends to the post-quantum setting, and is used as a building block to obtain our parallel repetition theorem as well.

Roughly speaking, given a prover whose success probability exceeds the desired bound $B(k, t, \eps)$, the lemma identifies a subset of $k'$ executions and threshold $t'$ such that the prover's success probability on these subverifiers exceeds $B(k', t', \eps)$ and is minimal in the sense that no further execution can be removed while preserving the corresponding guarantee. 

As in \cite{HS11,Chu11}, we will use the lemma to reduce the proof to two cases: when $k' = 1$, we immediately obtain a successful single-copy prover, while when $k' > 1$, the resulting minimality property provides the structure needed for one of the main quantum reductions.
\begin{lemma}[Adapted from \cite{HS11,Chu11}]\label{lemma:corr_reduction}
    Let $(P, V)$ be a $m$-round interactive protocol where $m = \poly(\secp)$. There exists a transformation $\calT$ such that for all $t \leq k \in \bbN$, every $\nu, \eps \in (0, 1)$, and every adversary $\calB$ against $\widehat{V}^{k, t}$,
    \begin{enumerate}
        \item $\calT(1^{\secp}, t, k, \nu, \eps)$ runs in time $\poly(\secp, k, \nu^{-1})$ given oracle access to $\calB$ and outputs $t' \in [t], k' \in [k]$ and an adversary $\calA$ which runs in time $\poly(\secp, k)$ given oracle access to $\calB$. 
        \item If 
            \[ \Pr[\langle \calB, \widehat{V}^{k, k} \rangle (1^{\secp}) = 1] \geq B(k, t, \eps)+\nu, \]
        then with probability at least $1-2^{-\secp}$, the output $(t', k', \calA)$ of $\calT$ satisfies the following:
            \[ \Pr[\langle \calA, \widehat{V}^{k', t'}\rangle (1^{\secp}) = 1] \geq B(k', t', \eps) + \frac{10k'-1}{10k} \cdot \nu, \]
        and either $k' = 1$, or with probability at least $1-\nu/10k$ over $r^{*}$, 
            \[ \Pr_{(\tau, r) \gets (\calA, \widehat{V}^{k', t'})}\left[\sum_{j \neq 1} V_j(\tau_j, r_j) \geq t' \Bigg| r_{\Query, 1} = r^{*}\right] \leq B(k'-1, t', \eps)+\frac{10(k'-1)+1}{10k} \cdot \nu, \]
        and
            \[ \Pr_{(\tau, r) \gets (\calA, \widehat{V}^{k', t'})}\left[\sum_{j \neq 1} V_j(\tau_j, r_j) \geq t'-1 \Bigg| r_{\Query, 1} = r^{*}\right] \leq B(k'-1, t'-1, \eps)+\frac{10(k'-1)+1}{10k} \cdot \nu. \]
    \end{enumerate}
\end{lemma}
The proof of Lemma \ref{lemma:corr_reduction} was originally stated for $\PPT$ adversaries $\calB$ in the three-message setting but can be easily adapted to generic protocols and $\QPT$ adversaries as well. For the sake of completeness, we include a proof of Lemma \ref{lemma:corr_reduction} in the appendix.

\subsection{Singular Value Decomposition}
We start by recalling Jordan's lemma \cite{Jor75}, which characterizes the joint action of two projectors $\Pi_A$ and $\Pi_B$ on a Hilbert space $\calH$.

\begin{lemma}[\cite{Jor75}]
    For any two Hermitian projectors $\Pi_A$ and $\Pi_B$ on a Hilbert space $\calH$, there exists an orthogonal decomposition of $\calH = \bigoplus_j \calS_j$ into one-dimensional and two-dimensional subspaces $\{\calS_j\}_j$, where each $\calS_j$ is invariant under both $\Pi_A$ and $\Pi_B$. Moreover:
    \begin{itemize}
        \item in each one-dimensional space, $\Pi_A$ and $\Pi_B$ act as identity or rank-zero projectors; and
        \item in each two-dimensional subspace $\calS_j$, $\Pi_A$ and $\Pi_B$ are rank-one projectors. In particular, there exist distinct orthogonal bases $\{\ket*{v_j}, \ket*{v_j^{\perp}}\}$ and $\{\ket*{w_j}, \ket*{w_j^{\perp}}\}$ for $\calS_j$ such that $\Pi_A$ projects onto $\ket*{v_j}$ and $\Pi_B$ projects onto $\ket*{w_j}$. We call $\{\ket{v_j}\}_j$ and $\{\ket{w_j}\}_j$ the singular vectors of $\Pi_A$ and $\Pi_B$, respectively, and $\{\varsigma_j\}_j = \{\braket{w_j}{v_j}\}_j$ the associated singular values.
    \end{itemize}
    Note that we can view each one-dimensional subspace $\calS_j$ as a degenerate two-dimensional subspace. If $\Pi_A|_{\calS_j} = \id$ then we label the (unit) vector spanning the subspace $\ket{v_j}$ and the zero vector by $\ket*{v_j^{\perp}}$; otherwise, when $\Pi_A|_{\calS_j} = \zero$ we label the spanning vector $\ket*{v_j^{\perp}}$ and the zero vector $\ket{v_j}$. We do the same for $\Pi_B|_{\calS_j}$ and the vectors $\ket{w_j}$ and $\ket*{w_j^{\perp}}$.
\end{lemma}
In other words, we can write $\Pi_A = \sum_j \ketbra{v_j}{v_j}$ and $\Pi_B = \sum_j \ketbra{w_j}{w_j}$. This allows us to write $\Pi_A \Pi_B \Pi_A = \sum_j p_j \ketbra{v_j}$ and $\Pi_B \Pi_A \Pi_B = \sum_j p_j \ketbra{w_j}$, where $p_j = |\braket{w_j}{v_j}|^2$. We define $\Pi_{ABA, \gamma}$ and $\Pi_{BAB, \gamma}$ to be the projections onto $\mathsf{span}\{\ket{v_j}\}_{j: p_j \geq \gamma}$ and $\mathsf{span}\{\ket{w_j}\}_{j: p_j \geq \gamma}$ respectively.

We can also consider the corresponding singular value decompositions $\Pi_B \Pi_A = \sum_j \varsigma_j \ketbra{w_j}{v_j}$ and $\Pi_A \Pi_B = \sum_j \varsigma_j \ketbra{v_j}{w_j}$, where $\varsigma_j := \sqrt{p_j}$.

\subsubsection{Quantum Singular Value Transformation (QSVT)}
We will use the following algorithmic result, which essentially states that for any Hermitian projectors $\Pi_A$ and $\Pi_B$ and for every state in $\im(\Pi_A)$ (resp. $\im(\Pi_B)$) whose corresponding singular values are not too small, one can transform the state to one in $\im(\Pi_B)$ (resp. $\im(\Pi_A)$) while approximately preserving the relative weights between singular vectors.
\begin{definition}[$C_{\Pi}NOT$ gate]
    For a Hermitian projector $\Pi$ we define the $\Pi$-controlled NOT gate as the unitary operator $C_{\Pi}NOT := X \otimes \Pi + I \otimes (I-\Pi)$.
\end{definition}
\begin{theorem}[Modified from \cite{GSLW19}]\label{thm:qsvt}
    Let $\Pi_B, \Pi_A$ be Hermitian projectors on a Hilbert space $\calH$ and let $\eps, \delta \in (0, 1)$. Consider the singular value decomposition 
        \[ \Pi_B\Pi_A = \sum_j \varsigma_j \ketbra{w_j}{v_j}. \]
    Then there is a unitary $U^{\Pi_B, \Pi_A}_{\eps, \delta}$ such that 
        \[ \Pi_B U^{\Pi_B, \Pi_A}_{\eps, \delta} \Pi_A = \sum_j \tilde{\varsigma}_j \ketbra{w_j}{v_j}, \]
    where for all $j$, $\tilde{\varsigma}_j \in \bbR$ and 
        \[ \tilde{\varsigma}_j = \begin{cases} \in [1-\eps, 1] & \text{if } \varsigma_j \geq \delta, \\ \in [-1, 1] & \text{otherwise.} \end{cases} \]
    Moreover, $U^{\Pi_B, \Pi_A}_{\eps, \delta}$ can be implemented by a circuit $C^{\Pi_B, \Pi_A}_{\eps, \delta}$ with the following properties:
    \begin{enumerate}
        \item $C^{\Pi_B, \Pi_A}_{\eps, \delta}$ uses one ancilla qubit and consists of $m$ calls to $C_{\Pi_B}NOT$, $m$ calls to $C_{\Pi_A}NOT$, and $m$ single-qubit gates on the ancilla qubit, where $m = O\left(\frac{\log(1/\eps)}{\delta}\right)$.
        \item For all states $\ket{\psi}$, $C^{\Pi_B, \Pi_A}_{\eps, \delta} (\ket{\psi} \otimes \ket{0}) = \left(U^{\Pi_B, \Pi_A}_{\eps, \delta} \ket{\psi}\right) \otimes \ket{0}$. 
    \end{enumerate}
\end{theorem}

\begin{corollary}\label{cor:qsvt_threshold}
    Let $\Pi_B, \Pi_A$ be Hermitian projectors on a Hilbert space $\calH$ and let $\eps, \delta \in (0, 1)$. If $\ket{\psi} = \sum_j \alpha_j \ket{v_j} \in \im(\Pi_A)$ is a state such that $\sum_{j: \varsigma_j < \delta} |\alpha_j|^2 \leq p$, then 
        \[ \left\|(I-\Pi_B)U^{\Pi_B, \Pi_A}_{\eps, \delta} \ket{\psi}\right\|_2^2 \leq p+2\eps, \]
    where $U^{\Pi_B, \Pi}_{\eps, \delta}$ is the unitary defined by Theorem \ref{thm:qsvt}.
\end{corollary}
\begin{proof}
    Theorem \ref{thm:qsvt} implies that 
        \[ \Pi_B U^{\Pi_B, \Pi_A}_{\eps, \delta} \Pi_A \ket{\psi}=\sum_j \alpha_j\tilde{\varsigma}_j\ket{w_j}, \]
    which together with the fact that $\ket{\psi} \in \im(\Pi_A)$, implies that 
        \[ \Pi_B U^{\Pi_B, \Pi_A}_{\eps, \delta} \ket{\psi} = \sum_j \alpha_j \tilde{\varsigma}_j \ket{w_j}. \]
    Thus, we can decompose the state $\ket{\phi} := U^{\Pi_B, \Pi_A}_{\eps, \delta} \ket{\psi}$ as 
        \[ \ket{\phi} = \sum_j \gamma_j \ket*{w_j^{\perp}} + \sum_j \alpha_j \tilde{\varsigma}_j \ket{w_j}. \]
    Since $\ket{\phi}$ is a normalized state (as $U^{\Pi_B, \Pi_A}_{\eps, \delta}$ is unitary), we know that
        \[ \sum_j |\gamma_j|^2 + \sum_j |\alpha_j|^2 \tilde{\varsigma}_j^2 = 1. \]
    Thus,
    \begin{align*}
        \sum_j |\gamma_j|^2 &= 1 - \sum_j |\alpha_j|^2 \cdot \tilde{\varsigma}_j^2 \leq 1 - \sum_{j: \varsigma_j \geq \delta} |\alpha_j|^2 \cdot \tilde{\varsigma}_j^2 \leq 1 - (1-p)(1-\eps)^2 \leq p+2\eps.
    \end{align*}
\end{proof}

\section{On Oracle Distinguishers}\label{sec:oracle_distinguishers}
In this section we prove the lemma which will be the key workhorse in the analyses of the next section. We begin by introducing some definitions and notation.
\begin{definition}[\cite{LMS22}]
    Let $\ket{D_{b, \secp}} \propto \sum_{r_{\secp}} \ket{r_{\secp}}_{\sfR} \ket{D_{b, \secp}(r_{\secp})}_{\sfE}$ be states where $\{D_{0, \secp}\}_{\secp \in \bbN}$ and $\{D_{1, \secp}\}_{\secp \in \bbN}$ are efficiently sampleable families of distributions over classical strings with random coins $\{r_{\secp}\}_{\secp \in \bbN}$.\footnote{Note that efficient sampleability of $D_{b, \secp}$ implies the ability to efficiently prepare the associated quantum state $\ket{D_{b, \secp}}$.}

    We define the family of unitaries $U_b := \{U_{b, \secp}\}_{\secp \in \bbN}, U_b^{\dagger} := \{U_{b, \secp}^{\dagger}\}_{\secp \in \bbN}$ acting on $\sfO \otimes \sfR \otimes \sfE$ (where $\sfO$ is single-qubit register) as
        \begin{equation}\label{eqn:Ub} U_{b, \secp} := X_{\sfO} \otimes \ketbra{D_{b, \secp}}_{\sfR\sfE} + I_{\sfO} \otimes (I-\ketbra{D_{b, \secp}})_{\sfR\sfE}. \end{equation}
\end{definition}

Informally, the following lemma states that even with given oracle access to $U_b$, a distinguisher cannot tell if it has a copy of $\ket{D_0}$ or $\ket{D_1}$ (assuming the two distributions are indistinguishable). Intuitively, this comes from the fact that the only way for the distinguisher to find something meaningfully changed by $U_b$ is by applying it to $\ket{D_b}$ itself. However, since the register $\sfR$ is otherwise inaccessible, the distinguisher must use its single copy of $\ket{D_b}$, which looks the same with $U_b$ that $\ket{D_{1-b}}$ looks with $U_{1-b}$.

\begin{lemma}[\cite{LMS22}]\label{lemma:oracle_distinguisher}
    Suppose there exists a $\poly(\secp)$-sized quantum distinguisher $S^{U_b}$ (which on input $1^{\secp}$ only queries $U_{b, \secp}$) without direct access to $\sfR$\footnote{More explicitly, the distinguisher $S$ consists of adversarial steps that acts as identity on $\sfR$, interleaved with oracle calls which are applications of $U_b$ on $\sfR\sfE$.} achieving
        \[ \left|\Pr[S^{U_0}(1^{\secp}, \ket{D_{0, \secp}}_{\sfR\sfE}) = 1]-\Pr[S^{U_1}(1^{\secp}, \ket{D_{1, \secp}}_{\sfR\sfE}) = 1]\right| \geq \frac{1}{\poly(\secp)}. \]
    Then there exists a $\poly(\secp)$-sized quantum algorithm $T$ that on input $1^{\secp}$ distinguishes classical samples from the distributions $D_{0, \secp}$ and $D_{1, \secp}$.
\end{lemma}
\begin{remark}
    It is straightforward to extend Lemma \ref{lemma:oracle_distinguisher} to the infinitely-often and non-uniform poly-sized regime, and the resulting algorithm $T$ is non-uniform and succeeds only infinitely often (on the same input lengths).  
\end{remark}

We will use Lemma \ref{lemma:oracle_distinguisher} to prove the main technical result of this section. Before we give the complete statement of this result, we give a brief description of what we want to prove. In this setting, we consider a $\QPT$ adversary $\calA$ that takes as input the state $\ket{D_{b,\secp}}_{\sfR\sfE}$. $\calA$ may repeatedly apply arbitrary unitaries $V_i$ on register $\sfE$ and its own workspace, along with calls to the oracle $U_b$ which coherently tests whether the joint state in registers $\sfR\sfE$ is still in the original state $\ket{D_{b,\secp}}$. Our goal is roughly to show that if, at the end, the adversary has successfully ``returned'' the state in registers $\sfR\sfE$ back to the state $\ket{D_{b,\secp}}$, then its internal state is \emph{statistically} (rather than merely computationally) indistinguishable in the cases where $b = 0$ and $b = 1$.

\begin{lemma}\label{lemma:comp_to_stat}
    Let $\{D_{0, \secp}\}_{\secp \in \bbN}$ and $\{D_{1, \secp}\}_{\secp \in \bbN}$ be efficiently sampleable but computationally indistinguishable (to $\poly(\secp)$-size quantum adversaries with quantum advice) families of distributions over classical strings with associated states $\ket{D_b} = \{\ket{D_{b, \secp}}\}_{\secp \in \bbN}$. 
    
    For every family of $\poly(\secp)$-sized oracle circuits $\{V^{\calO_{\secp}}_{\secp}\}_{\secp \in \bbN}$,\footnote{This means that there exist polynomials $t$ and $s$ such that $V^{\calO_{\secp}}_{\secp}$ can be written as $(I_{\sfR} \otimes V_{t, \secp, \sfO\sfE\sfI}) (I_{\sfI} \otimes \calO_{\secp, \sfO\sfR\sfE}) \cdots (I_{\sfI} \otimes \calO_{\secp, \sfO\sfR\sfE})(I_{\sfR} \otimes V_{0, \secp, \sfO\sfE\sfI})$  where $V_{i, \secp, \sfO\sfE\sfI}$ is a unitary of size at most $s$.} and family of $\poly(\secp)$-sized (normalized) pure states $\{\ket{\phi_{\secp}}_{\sfI}\}_{\secp \in \bbN}$, let
        \[ \ket*{\psi_{b, V_{\secp}}}_{\sfO\sfR\sfE\sfI} := V^{U_{b, \secp}}_{\secp} \ket{0}_{\sfO}\ket{D_{b, \secp}}_{\sfR\sfE}\ket{\phi_{\secp}}_{\sfI}, \]
    where $U_b$ is defined in Equation \ref{eqn:Ub}, and 
        \[ \rho_{b, \secp} := \Tr_{\sfR\sfE}(\ketbra*{D_{b, \secp}}_{\sfR\sfE} \cdot \ketbra*{\psi_{b, V_{\secp}}}_{\sfO\sfR\sfE\sfI} \cdot \ketbra*{D_{b, \secp}}_{\sfR\sfE}). \]
        Then 
    \begin{align*}
        &\|(I-\ketbra*{D_{0, \secp}})_{\sfR\sfE} \ket{\psi_{0, V_{\secp}}}\|^2 \cdot \ketbra{0} \otimes \ketbra{\bot} + \ketbra{1} \otimes \rho_{0, \secp} \\
        \stackrel{\negl}\equiv& \\
        &\|(I-\ketbra*{D_{1, \secp}})_{\sfR\sfE} \ket{\psi_{1, V_{\secp}}}\|^2 \cdot \ketbra{0} \otimes \ketbra{\bot} + \ketbra{1} \otimes \rho_{1, \secp}.
    \end{align*}
\end{lemma}

Before diving into the formal proof of Lemma \ref{lemma:comp_to_stat}, we start by giving a high-level overview of the proof strategy. At an intuitive level, we proceed by contradiction. For each \(b \in \{0, 1\}\), let \(p_b\) denote the probability that the final projection onto \(\lvert D_b\rangle\) succeeds, and let $\ket{\upsilon_b}$ denote the normalized residual state conditioned on success. The final flagged state is therefore
    \[ \sigma_b = (1-p_b) \ketbra{0} \otimes \ketbra{\bot} + p_b \ketbra{1} \otimes \ketbra*{\upsilon_b}. \]
Because the success and failure branches have orthogonal classical flags, by the triangle inequality we have that
    \[ \TD(\sigma_0, \sigma_1) \leq 2 \lvert p_0-p_1 \rvert + \min\{p_0,p_1\} \cdot \TD\left(\ketbra*{\upsilon_0}, \ketbra*{\upsilon_1}\right). \]
Consequently, if the two final states were noticeably far apart, then either the success probabilities \(p_0\) and \(p_1\) would differ noticeably, or the successful residual states would differ noticeably after accounting for the probability with which the successful branch occurs.

In the first case, one can contradict Lemma \ref{lemma:oracle_distinguisher} by constructing a distinguisher that on input $\ket{D_{b,\secp}}$ and with oracle access to
$U_{b,\secp}$, runs $V_\secp^{U_{b,\secp}}$, tests whether the final state in $\sfR \sfE$ is $\ket{D_b}$, and outputs the result. Its acceptance probability is exactly $p_b$, giving rise to a noticeable gap $|p_0-p_1|$. In the second case, we can also construct a distinguisher $S^{U_b}$ for \(\ket{D_0}\) and \(\ket{D_1}\). Specifically, $S^{U_b}$ runs a reference experiment using \(D_0\) and $U_0$ alongside the experiment using the unknown distribution \(D_b\). If both final projections succeed, then it performs a SWAP test on the two residual states. If \(b=0\), the two residual states are identical, so the SWAP test accepts with probability \(1\). On the other hand, if \(b=1\), the SWAP test rejects with probability
    \[ \frac{p_0p_1}{2} \cdot \TD(\ketbra{\upsilon_0}, \ketbra{\upsilon_1})^2, \]
which is non-negligible whenever the weighted distance above is non-negligible.  
\begin{proof}[Proof of Lemma \ref{lemma:comp_to_stat}]
    We define the probabilities
        \[ p_{b, \secp} := \left\|\Tr_{\sfR\sfE}\left(\ketbra*{D_{b, \secp}}_{\sfR\sfE} \cdot \ketbra{\psi_{b, V_{\secp}}}_{\sfO\sfR\sfE\sfI} \cdot \ketbra*{D_{b, \secp}}_{\sfR\sfE}\right)\right\|^2. \]
    We also note that $\rho_{b, \secp} = p_{b, \secp} \ketbra*{\upsilon_{b, \secp}}$ for some pure state $\ket{\upsilon_{b, \secp}}$. 
    
    The lemma is then equivalent to the claim that
    \begin{align*}
        &\TD\Bigl(
            (1-p_{0, \secp}) \cdot \ketbra{0} \otimes \ketbra{\bot}
            + p_{0, \secp} \cdot \ketbra{1} \otimes \ketbra{\upsilon_{0, \secp}},
            \\
        &\qquad
            (1-p_{1, \secp}) \cdot \ketbra{0} \otimes \ketbra{\bot}
            + p_{1, \secp} \cdot \ketbra{1} \otimes \ketbra{\upsilon_{1, \secp}}
            \Bigr)
        \leq \negl(\secp).
    \end{align*}
        
    Suppose for the sake of contradiction that there exists a polynomial $p(\secp)$ such that for infinitely many $\secp$ (call this set $\Lambda$),
    \begin{align*}
        &\TD\Bigl(
            (1-p_{0, \secp}) \cdot \ketbra{0} \otimes \ketbra{\bot}
            + p_{0, \secp} \cdot \ketbra{1} \otimes \ketbra{\upsilon_{0, \secp}},
            \\
        &\qquad
            (1-p_{1, \secp}) \cdot \ketbra{0} \otimes \ketbra{\bot}
            + p_{1, \secp} \cdot \ketbra{1} \otimes \ketbra{\upsilon_{1, \secp}}
            \Bigr)
        \geq \frac{1}{p(\secp)}.
    \end{align*}
    By the triangle inequality, we have that
    \begin{align*}
        \frac{1}{p(\secp)}
        &\leq \TD\Bigl(
            (1-p_{0, \secp}) \cdot \ketbra{0} \otimes \ketbra{\bot}
            + p_{0, \secp} \cdot \ketbra{1} \otimes \ketbra{\upsilon_{0, \secp}},
            \\
        &\qquad
            (1-p_{1, \secp}) \cdot \ketbra{0} \otimes \ketbra{\bot}
            + p_{1, \secp} \cdot \ketbra{1} \otimes \ketbra{\upsilon_{1, \secp}}
            \Bigr) \\
        &\leq \TD\Bigl(
            (1-p_{0, \secp}) \cdot \ketbra{0} \otimes \ketbra{\bot},
            (1-p_{1, \secp}) \cdot \ketbra{0} \otimes \ketbra{\bot}
            \Bigr) \\
        &\qquad
            + \TD\Bigl(
            p_{0, \secp} \cdot \ketbra{1} \otimes \ketbra{\upsilon_{0, \secp}},
            p_{1, \secp} \cdot \ketbra{1} \otimes \ketbra{\upsilon_{1, \secp}}
            \Bigr) \\
        &= |(1-p_{0, \secp}) - (1-p_{1, \secp})|
           + \TD\left(p_{0, \secp} \cdot \ketbra{\upsilon_{0, \secp}},
                 p_{1, \secp} \cdot \ketbra{\upsilon_{1, \secp}}\right) \\
        &\leq |(1-p_{0, \secp}) - (1-p_{1, \secp})|
           + |p_{0, \secp} - p_{1, \secp}|
           + \min_{b \in \{0, 1\}} p_{b, \secp} \cdot \TD\left(\ket{\upsilon_{0, \secp}}, \ket{\upsilon_{1, \secp}}\right) \\
        &= 2|p_{0, \secp} - p_{1, \secp}|
           + \min_{b \in \{0, 1\}} p_{b, \secp} \cdot \TD\left(\ket{\upsilon_{0, \secp}}, \ket{\upsilon_{1, \secp}}\right).
    \end{align*}
    Suppose that $2|p_{0, \secp}-p_{1, \secp}| \geq \frac{1}{2p(\secp)}$ for infinitely many $\secp \in \Lambda$. We now use the fact that
    \begin{align*}
        &\Tr_{\sfR\sfE}(\ketbra{1}_{\sfB} \cdot \SWAP_{\sfB\sfO} \cdot U_{b, \secp, \sfO\sfR\sfE} \cdot \SWAP_{\sfB\sfO} \ketbra{0}_{\sfB}\ketbra*{\psi_{b, V_{\secp}}}_{\sfO\sfR\sfE\sfI} \cdot \SWAP_{\sfB\sfO} \cdot U^{\dagger}_{b, \secp, \sfO\sfR\sfE} \cdot \SWAP_{\sfB\sfO} \cdot \ketbra{1}_{\sfB}) \\
        =& \ketbra{1}_{\sfB} \cdot \Tr_{\sfR\sfE}(\SWAP_{\sfB\sfO} \cdot U_{b, \secp, \sfO\sfR\sfE} \cdot \SWAP_{\sfB\sfO} \ketbra{0}_{\sfB}\ketbra*{\psi_{b, V_{\secp}}}_{\sfO\sfR\sfE\sfI} \cdot \SWAP_{\sfB\sfO} \cdot U^{\dagger}_{b, \secp, \sfO\sfR\sfE} \cdot \SWAP_{\sfB\sfO}) \\
        =& \ketbra{1}_{\sfB} \otimes \Tr_{\sfR\sfE}(\ketbra*{D_{b, \secp}}_{\sfR\sfE} \cdot \ketbra*{\psi_{b, V_{\secp}}}_{\sfO\sfR\sfE\sfI} \cdot \ketbra*{D_{b, \secp}}_{\sfR\sfE}).
    \end{align*}
    We consider the non-uniform oracle algorithm $S^{U_{b, \secp}}$ with advice $\ket{\phi_{\secp}}_{\sfI}$ which on input $(1^{\secp}, \ket{D_{b, \secp}}_{\sfR\sfE})$, appends $\ket{0}_{\sfB\sfO}$ before applying 
        \[ \SWAP_{\sfB\sfO} \cdot U_{b, \secp, \sfO\sfR\sfE} \cdot \SWAP_{\sfB\sfO} \cdot V^{U_{b, \secp, \sfO\sfR\sfE}}_{\secp, \sfO\sfR\sfE\sfI} \cdot \ket{0}_{\sfB}\ket{0}_{\sfO}\ket{D_{b, \secp}}_{\sfR\sfE}\ket{\phi_{\secp}}_{\sfI} \]
    and then measuring the register $\sfB$ and outputting the measurement outcome. We have that
        \[ \Pr[S^{U_{b, \secp}}(1^{\secp}, \ket{D_{b, \secp}}_{\sfR\sfE}, \ket{\phi_{\secp}}_{\sfI}) = 1] = p_{b, \secp}, \]
    and thus
        \[ \left|\Pr[S^{U_{0, \secp}}(1^{\secp}, \ket{D_{0, \secp}}_{\sfR\sfE}, \ket{\phi_{\secp}}_{\sfI}) = 1]-\Pr[S^{U_{1, \secp}}(1^{\secp}, \ket{D_{1, \secp}}_{\sfR\sfE}, \ket{\phi_{\secp}}_{\sfI}) = 1]\right| \geq |p_{0, \secp}-p_{1, \secp}| \geq \frac{1}{4p(\secp)} \]
    for infinitely many $\secp$, contradicting the indistinguishability of $D_0$ and $D_1$ by Lemma \ref{lemma:oracle_distinguisher}.
    
    As $2|p_{0, \secp}-p_{1, \secp}| < \frac{1}{2p(\secp)}$ for all but finitely many $\secp \in \Lambda$, it must be the case that 
        \[ \min_{b \in \{0, 1\}} p_{b, \secp} \cdot \TD\left(\ket{\upsilon_{0, \secp}}, \ket{\upsilon_{1, \secp}}\right) \geq \frac{1}{2p(\secp)}, \]
    for infinitely many $\secp$ (call this set $\Lambda'$).
    
    We now consider the following circuit $C^{U_{b, \secp, \sfO\sfR\sfE}}(1^{\secp}, \ket{D_{b, \secp}}_{\sfR\sfE}, \ket{\phi_{\secp}}_{\sfI'} \otimes \ket{\phi_{\secp}}_{\sfI})$:
    \begin{enumerate}
        \item Append the state $\ket{0, 0}_{\sfB'\sfO'}\ket{D_{0, \secp}}_{\sfR'\sfE'}\ket{0, 0}_{\sfB\sfO}$ to the state in registers $\sfR\sfE\sfI\sfI'$.
        \item Apply the unitaries
            \[ U_{\mathsf{Control}, \secp} := \SWAP_{\sfB'\sfO'} \cdot U_{0, \secp, \sfO'\sfR'\sfE'} \cdot \SWAP_{\sfB'\sfO'} \cdot V^{U_{0, \secp, \sfO'\sfR'\sfE'}}_{\secp, \sfO'\sfR'\sfE'\sfI'} \]
        and
            \[ U_{\mathsf{Test}, b, \secp} := \SWAP_{\sfB\sfO} \cdot U_{b, \secp, \sfO\sfR\sfE} \cdot \SWAP_{\sfB\sfO} \cdot V^{U_{b, \secp, \sfO\sfR\sfE}}_{\secp, \sfO\sfR\sfE\sfI}. \]
        \item Measure registers $\sfB$ and $\sfB'$. If both output 1, then run a SWAP test on the states on the registers $\sfO\sfE\sfI$ and $\sfO'\sfE'\sfI'$ and output the result. Otherwise, output 1.
    \end{enumerate}
    The efficiency of sampling $D_{b, \secp}$ implies that we can create $\ket{D_{0, \secp}}$ and implement $U_{0, \secp, \sfO'\sfR'\sfE'}$ efficiently. Since $\{V_{i, \secp}\}_{i = 0}^t$ are $\poly(\secp)$-sized and $t = \poly(\secp)$, $C$ is efficient.

    We now observe that $U_{\mathsf{Control}, \secp}$ and $U_{\mathsf{Test}, b, \secp}$ commute as they act on different registers and that
        \[ U_{\mathsf{Test}, b, \secp}\ket{0}_{\sfB\sfO}\ket{D_{b, \secp}}_{\sfR\sfE}\ket{\phi_{\secp}}_{\sfI} = \sqrt{1-p_{b, \secp}} \cdot \ket{0}_{\sfB} \otimes \ket{\mu_{b, \secp}} + \sqrt{p_{b, \secp}} \cdot \ket{1}_{\sfB} \otimes  \ket{\upsilon_{b, \secp}} \]
    for some (normalized) pure state $\ket{\mu_{b, \secp}}$, so
    \begin{align*}
        &\Pr[C^{U_{b, \secp}}(1^{\secp}, \ket{D_{b, \secp}}, \ket{\phi_{\secp}}_{\sfI'} \otimes \ket{\phi_{\secp}}_{\sfI}) = 1] \\
        = \,\, &(1 - p_{0, \secp} p_{b, \secp}) \cdot 1 + p_{0, \secp} p_{b, \secp} \cdot \left[1-\frac{1}{2}\TD\left(\ket{\upsilon_{0, \secp}}, \ket{\upsilon_{b, \secp}}\right)^2  \right] \\
        = \,\, &1 - \frac{p_{0, \secp}p_{b, \secp}}{2} \cdot \TD\left(\ket{\upsilon_{0, \secp}}, \ket{\upsilon_{b, \secp}}\right)^2.
    \end{align*}
    Therefore, for all $\secp \in \Lambda'$,
        \[ \Pr[C^{U_{0, \secp}}(1^{\secp}, \ket{D_{0, \secp}}, \ket{\phi_{\secp}}_{\sfI'} \otimes \ket{\phi_{\secp}}_{\sfI}) = 1] = 1, \]
    while
    \begin{align*}
        \Pr[C^{U_{1, \secp}}(1^{\secp}, \ket{D_{1, \secp}}, \ket{\phi_{\secp}}_{\sfI'} \otimes \ket{\phi_{\secp}}_{\sfI}) = 1] &= 1 - \frac{p_{0, \secp}p_{1, \secp}}{2} \cdot \TD\left(\ket{\upsilon_{0, \secp}}, \ket{\upsilon_{1, \secp}}\right)^2 \\
        &\leq 1 - \frac{1}{2} \cdot \left[\min_{b \in \{0, 1\}} p_{b, \secp} \cdot \TD\left(\ket{\upsilon_{0, \secp}}, \ket{\upsilon_{1, \secp}}\right)\right]^2 \\
        &\leq 1-\frac{1}{8p(\secp)^2}.
    \end{align*}
    Therefore,
        \[ \left|\Pr[C^{U_{0, \secp}}(1^{\secp}, \ket{D_{0, \secp}}, \ket{\phi_{\secp}}_{\sfI'} \otimes \ket{\phi_{\secp}}_{\sfI}) = 1]-\Pr[C^{U_{1, \secp}}(1^{\secp}, \ket{D_{1, \secp}}, \ket{\phi_{\secp}}_{\sfI'} \otimes \ket{\phi_{\secp}}_{\sfI}) = 1]\right| \geq \frac{1}{8p(\secp)^2} \]
    for infinitely many $\secp$ (in particular, for all $\secp \in \Lambda'$), which contradicts the indistinguishability of $D_0$ and $D_1$ by Lemma \ref{lemma:oracle_distinguisher}.

    We conclude that in fact
    \begin{align*}
        &\|(I-\ketbra*{D_{0, \secp}})_{\sfR\sfE} \ket{\psi_{0, V_{\secp}}}\|^2 \cdot \ketbra{0} \otimes \ketbra{\bot} + \ketbra{1} \otimes \rho_{0, \secp} \\
        \stackrel{\negl}\equiv& \\
        &\|(I-\ketbra*{D_{1, \secp}})_{\sfR\sfE} \ket{\psi_{1, V_{\secp}}}\|^2 \cdot \ketbra{0} \otimes \ketbra{\bot} + \ketbra{1} \otimes \rho_{1, \secp},
    \end{align*}
    as desired.
\end{proof}

\section{Parallel Repetition for Private-Coin Protocols}\label{sec:private-coin}
\begin{theorem}\label{thm:poly_round_tight} 
    Let $(P, V)$ be a $m$-round interactive protocol, where $m(\secp) = \poly(\secp)$. Suppose that there exists a $\QPT$ adversary $\calB$, polynomials $k(\secp) \geq k'(\secp) \geq t'(\secp)$ such that $k'(\secp) > 1$, inverse polynomial $\nu(\secp) > 0$, and function $\eps(\secp) \geq 0$ such that for all $\secp \in \Lambda \subseteq \bbN$,
        \[ \Pr[\langle \calB, \widehat{V}^{k', t'}\rangle (1^{\secp}) = 1] \geq B(k', t', \eps) + \frac{10k'-1}{10k} \cdot \nu, \]
    and with probability at least $1-\nu/10k$ over $r^{*}$, 
        \[ \Pr_{(\tau, r) \gets (\calB, \widehat{V}^{k', t'})}\left[\sum_{j \neq 1} V_j(\tau_j, r_j) \geq t' \Bigg| r_{\Query, 1} = r^{*}\right] \leq B(k'-1, t', \eps)+\frac{10(k'-1)+1}{10k} \cdot \nu, \]
    and
        \[ \Pr_{(\tau, r) \gets (\calB, \widehat{V}^{k', t'})}\left[\sum_{j \neq 1} V_j(\tau_j, r_j) \geq t'-1 \Bigg| r_{\Query, 1} = r^{*}\right] \leq B(k'-1, t'-1, \eps)+\frac{10(k'-1)+1}{10k} \cdot \nu. \]
    Then there exists a $\QPT$ adversary $\calA$ and negligible function $\negl(\secp)$ such that for all $\secp \in \Lambda$,
    \[
        \Pr[\langle \calA, V \rangle (1^\secp) = 1] \geq \eps+\nu/2k-\negl(\secp).
    \]
\end{theorem}

\begin{remark}
    Our proofs are nearly identical in both the uniform and non-uniform setting, although in the non-uniform setting our reduction requires more copies of the prover's initial state, which is likely necessary (see \cite{BQSY24}). For the sake of simplicity we therefore assume all adversaries have quantum non-uniformity (although it is straightforward to verify that our reduction is uniformity-preserving).
\end{remark}

\begin{corollary}[Reminder of Theorem \ref{thm:tight_repetition}]\label{cor:fhe_tight_parallel_repetition}
    Let $(P, V)$ be a $m = \poly(\secp)$-round interactive protocol and let $\eps = \eps(\secp)$ be any parameter such that there exists a negligible function $\negl$ where for every $\QPT$ adversary $\calA$ and all $\secp \in \bbN$,
        \[ \Pr[\langle \calA, V \rangle (1^\secp) = 1] \leq \eps+\negl(\secp). \]
    Fix any polynomials $k := k(\secp)$ and $t := t(\secp) \leq k$. Then there exists a negligible function $\negl'$ such that for every $\QPT$ adversary $\calB$,
        \[ \Pr[\langle \calB, \widehat{V}^{k, t} \rangle (1^\secp) = 1] \leq B(k, t, \eps) + \negl'(\secp) \]
    for all $\secp \in \bbN$.
\end{corollary}
\begin{proof}[Proof of Corollary \ref{cor:fhe_tight_parallel_repetition}]
    The statement is trivial for $k = 1$, so we assume $k > 1$ for the rest of the analysis. Suppose for the sake of contradiction that there exists a $\QPT$ adversary $\calB$, inverse polynomial $\nu(\secp) > 0$, and infinite set $\Lambda \subseteq \bbN$ such that
        \[ \Pr[\langle \calB, \widehat{V}^{k, k} \rangle (1^\secp) = 1] \geq B(k, t, \eps) + \nu \]
    for all $\secp \in \Lambda$.

    Running the transformation $\calT$ from Lemma \ref{lemma:corr_reduction} gives $(k', t', \calA')$ in time $\poly(\secp, k, \nu^{-1}) = \poly(\secp)$ such that with probability at least $1-2^{-\secp}$, 
        \[ \Pr[\langle \calA', \widehat{V}^{k', t'}\rangle (1^{\secp}) = 1] \geq B(k', t', \eps) + \frac{10k'-1}{10k} \cdot \nu, \]
    and either $k' = 1$, or with probability at least $1-\nu/10k$ over $r^{*}$, 
        \[ \Pr_{(\tau, r) \gets (\calA', \widehat{V}^{k', t'})}\left[\sum_{j \neq 1} V_j(\tau_j, r_j) \geq t' \Bigg| r_{\Query, 1} = r^{*}\right] \leq B(k'-1, t', \eps)+\frac{10(k'-1)+1}{10k} \cdot \nu, \]
    and
        \[ \Pr_{(\tau, r) \gets (\calA', \widehat{V}^{k', t'})}\left[\sum_{j \neq 1} V_j(\tau_j, r_j) \geq t'-1 \Bigg| r_{\Query, 1} = r^{*}\right] \leq B(k'-1, t'-1, \eps)+\frac{10(k'-1)+1}{10k} \cdot \nu. \]
    Moving forward we assume that $\calT$ always succeeds since it only fails with probability $2^{-\secp} = \negl(\secp)$. If $k' = 1$ and $t' = 1$, then 
        \[ \Pr[\langle \calA', \widehat{V} \rangle (1^{\secp}) = 1] \geq B(1, 1, \eps) + \frac{9\nu}{10k} = \eps+\frac{9\nu}{10k}, \]
    for all $\secp \in \Lambda$, contradicting the soundness of $(\widehat{P}, \widehat{V})$ by our assumption and Lemma \ref{lemma:qfhe_transform}. Similarly, if $k' = 1$ and $t' = 0$ then 
        \[ \Pr[\langle \calA', \widehat{V} \rangle (1^{\secp}) = 1] \geq B(1, 0, \eps) + \frac{9\nu}{10k} > 1, \]
    which is impossible.

    Otherwise, we have some $k' > 1$ and $\QPT$ adversary $\calA'$ such that for all $\secp \in \Lambda$, with probability at least $1-\nu/10k$ over $r^{*}$, 
        \[ \Pr_{(\tau, r) \gets (\calA', \widehat{V}^{k', t'})}\left[\sum_{j \neq 1} V_j(\tau_j, r_j) \geq t' \Bigg| r_{\Query, 1} = r^{*}\right] \leq B(k'-1, t', \eps)+\frac{10(k'-1)+1}{10k} \cdot \nu, \]
    and
        \[ \Pr_{(\tau, r) \gets (\calA', \widehat{V}^{k', t'})}\left[\sum_{j \neq 1} V_j(\tau_j, r_j) \geq t'-1 \Bigg| r_{\Query, 1} = r^{*}\right] \leq B(k'-1, t'-1, \eps)+\frac{10(k'-1)+1}{10k} \cdot \nu. \]
    
    Applying Theorem \ref{thm:poly_round_tight} gives a $\QPT$ adversary $\calA$ and negligible function $\negl$ such that for all $\secp \in \Lambda$,
    \begin{align*}
        \Pr[\langle \calA, V \rangle (1^\secp) = 1] &\geq \eps+\nu/2k-\negl(\secp),
    \end{align*}
    contradicting the soundness of $(P, V)$.
    
    We conclude that $(\widehat{P}^{\otimes k}, \widehat{V}^{k, t})$ has soundness $B(k, t, \eps)+\negl(\secp)$ as desired.
\end{proof}

\begin{proof}[Proof of Theorem \ref{thm:poly_round_tight}]
Denote by $\ket{\psi}$ the initial state of $\calB$. For the sake of simplicity, we assume without loss of generality that $\ket{\psi}$ contains a designated register $\sfZ$ composed of subregisters $\sfZ_1, \ldots, \sfZ_m$, where $\sfZ_j$ is the register on which $\calB$ writes its response to the $j$'th round queries. We denote all the other registers in $\ket{\psi}$ by $\sfI$.

For every $j \in [m]$ we denote by $U_j$ the collection of unitaries indexed by the $j$th round queries $\bar{q}_j = (q_{j,1}, \ldots, q_{j,k})$ such that $U_j(\bar{q}_j)$ is the unitary that $\calB$ applies to its quantum state upon receiving query $\bar{q}_j$ (which operates on $\sfI$ and $\sfZ_j$). The assumption that each $U_j(\bar{q}_{j})$ is a unitary is without loss of generality, as any adversary not of this form can be ``purified'' into an adversary with unitary strategies, identical observable behavior, and constant factor blowup in size. 

Before introducing the details of our reduction, we give some intuition for the high-level strategy. The reduction embeds the external verifier in the first coordinate and emulates the verifier in all the other coordinates $[k] \setminus \{1\}$. Before receiving any queries from the external verifier, $\calA$ runs the $k$-fold prover $\calB$ using a dummy encryptions of $0$ in the first coordinate and post-selects on exactly $t'-1$ coordinates outside of the first accepting. The QSVT algorithm  plays the role of quantum rewinding: By Corollary \ref{cor:qsvt_threshold}, it coherently transfers the post-selected state from the acceptance subspace back to the dummy-encryption subspace, allowing the encryption registers to be removed while retaining the residual prover state. Once the first query $q_1$ is received, the reduction appends an encryption of $q_1$ and applies the corresponding inverse QSVT transformation to return to the acceptance subspace. Crucially, rather than fixing the randomness used for the key-generation and the ciphertext generation, the dummy ciphertext is prepared as a coherent purification, denoted by $\ket{\Enc(0)}_{\sfR\sfE}$ in the reduction, where $\sfR$ is the register containing the randomness. Since the prover has no access to the purifying register $\sfR$, its view is exactly the same as if these coins had been sampled classically. Retaining them coherently, however, allows the reduction to reflect about the entire encryption state and places us in the setting of Lemma \ref{lemma:comp_to_stat}, which justifies replacing $\ket{\Enc(0)}$ by $\ket{\Enc(q_1)}$ even after the post-selection. If the coins were fixed at the outset, this would instead require a pointwise replacement of $\Enc(0;r)$ by $\Enc(q_1;r)$ for a fixed $r$, whereas semantic security only guarantees indistinguishability over the random choice of the encryption coins. We can then apply a new set of rotations for each new query that is received using the same logic. The remaining verifier randomness may be sampled classically because it is never subject to this post-selection-and-replacement step. 

We now formally describe our reduction, which takes as input $N := 20k\lambda/\nu$ copies of $\ket{\psi}$.\footnote{Throughout the description of the reduction, we will often assume that any values which have been sampled or measured are automatically incorporated in the future. For example, in step 2(c), we assume that $\pk_{j, \ell}$ and $\sk_{j, \ell}$ for all $j \neq 1$ are computed using $r_{\Ver}$; similarly, in steps 3(b)/(e), the state $\ket*{\Enc(0)^{\rnd+1}}$ implicitly depends on $\pk_{1, \rnd}$, whose value is measured earlier in step 3(d) (and thereby fixed in step 3(b)/(e)).}

\algo{The Reduction $\calA$}{A-poly-desc}{
$\calA$ begins with state $\ket{\psi}^{\otimes N}$, and runs as follows:
\begin{enumerate}
    \item Set $\eps = 2^{-\secp}$, $\mu = \nu/10km$, $\delta =  \mu^2/2N^2$, $\tau_0 := \emptyset$ and $\widehat{\sk}_{[k], 0} := \emptyset$.
    \item \label{a1loop} For $s = 1$ to $N$:
    \begin{enumerate}
        \item Sample the following sets of randomness: 
        \begin{align*}
            \{r_{\Gen, j, \ell}, r_{\Enc, j, \ell}\}_{j \neq 1, \ell \in [m]}, \{r'_{\Enc, j, \ell}\}_{j \neq 1, 2 \leq \ell \leq m}, \{r_{\Query, j}\}_{j \neq 1}.
        \end{align*}
        We refer to all the randomness sampled in this step by $r_{\Ver}$.
        \item Append the state
            \[ \ket*{\Enc(0)^1}_{\sfR\sfE} \propto \sum_{r_{\Gen, 1}, r_{\Enc, 1}} \ket*{r_{\Gen, 1}, r_{\Enc, 1}}_{\sfR} \ket*{\pk_{1, 1}, \widehat{0}^1, \pk_{1, 2}, \widehat{\sk}_{1, 1}, \widehat{0}^2, \ldots, \pk_{1, m}, \widehat{\sk}_{1, m-1}, \widehat{0}^m}_{\sfE} \]
        which is constructed as follows: first append 
            \[ \ket*{+^1} \propto \sum \ket*{r_{\Gen, 1, 1}, \ldots, r_{\Gen, 1, m}, r_{\Enc, 1, 1}, r_{\Enc, 1, 2}, r'_{\Enc, 1, 2}, \ldots, r_{\Enc, 1, m}, r'_{\Enc, 1, m}} \]
        in the $\sfR$ register (where $r_{\Gen, 1, \ell}$ and $r_{\Enc, 1, \ell}$, as well as $r'_{\Enc, 1, \ell}$ if $\ell > 1$, are in subregister $\sfR_{\ell}$). Then, in subregister $\sfE_{\ell}$, coherently compute 
            \[ \pk_{1, \ell} \gets \QHE.\Gen(1^{\secp}; r_{\Gen, 1, \ell}), \]
            and
            \[ \widehat{0}^{\ell} = \QHE.\Enc(\pk_{1, \ell}, 0; r_{\Enc, 1, \ell}) \]
        for $1 \leq \ell \leq m$. For $2 \leq \ell \leq m$, additionally coherently compute in $\sfE_{\ell}$ the value
            \[ \widehat{\sk}_{1, \ell-1} = \QHE.\Enc(\pk_{1, \ell}, \sk_{1, \ell-1}; r'_{\Enc, 1, \ell}). \]
        \item Define the projection 
            \[ \Pi_{\Ver} := \sum_{\tau: \sum_{i \neq 1} V_i(\tau_j, r_{\Ver, j}) = t'-1} \ketbra*{\tau}_{\sfE\sfZ}. \]
        
        Define the projection $\Pi_{\Acc, \rnd} := U_{\Ver, \rnd}^{\dagger} \ldots U_{\Ver, m}^{\dagger} \Pi_{\Ver} U_{\Ver, m} \ldots U_{\Ver, \rnd}$, where $U_{\Ver, \ell}$ is the following unitary:
        \begin{enumerate}
            \item Controlled on the partial transcript in registers $\sfE_{\rnd} \ldots \sfE_{\ell}$ and $\sfZ_{\rnd} \ldots \sfZ_{\ell-1}$, apply the prover unitary $U_{\ell, \sfZ_{\ell}\sfI}$.
            
            For example, if $\ell = 1$, then controlled on $\pk_{1, 1}, \ct_{1, 1}$ in register $\sfE_1$, apply $U_1(\pk_{[k], 1} || \widehat{q}_{1, 1} || \cdots || \ct_{1, 1} || \cdots || \widehat{q}_{k, 1})$, where $\pk_{-1, 1}, \widehat{q}_{-1, 1}$ are computed with $r_{\Ver}$.
        \end{enumerate}
        Measure $\Pi_{\Acc, 1}$. If the outcome is 0, do the following. If $s < N$, proceed to the next iteration of step~\ref{a1loop}. If $s = N$ then abort the algorithm and output $\bot$.
        \item Apply the circuit $C^{\ketbra*{\Enc(0)^1}, \Pi_{\Acc, 1}}_{\eps, \delta}$ with ancilla qubit $\ket{0}_{\sfO}$.
        \item Measure the projection $\ketbra*{\Enc(0)^1}_{\sfR\sfE}$ before discarding registers $\sfR$ and $\sfE$.
        \item If the outcome of the projection is 0, discard registers $\sfO\sfZ\sfI$; return to the start of step~\ref{a1loop} if $s < N$ and abort/output $\bot$ if $s = N$.
    \end{enumerate}
    \item For $\rnd = 1$ to $m$:
    \begin{enumerate}
        \item Receive the $(\rnd)$-th query $q_{\rnd}$ from the external verifier and set $q_{1, \rnd} = q_{\rnd}$.
        \item Append the state $\ket*{\Enc(q_{1, \rnd})^{\rnd}}_{\sfR_{\rnd}\ldots\sfR_m\sfE_{\rnd}\ldots\sfE_m}$, which is computed as follows.

        For $\rnd \leq \ell \leq m$, if $\ell = 1$, append 
            \[ \ket*{+^{\rnd}_1} \propto \sum_{r_{\Gen, 1, 1}, r_{\Enc, 1, 1}} \ket*{r_{\Gen, 1, 1}, r_{\Enc, 1, 1}} \]
        to register $\sfR_1$; otherwise, append 
            \[ \ket*{+^{\rnd}_{\ell}} \propto \sum_{r_{\Gen, 1, \ell}, r_{\Enc, 1, \ell}, r'_{\Enc, 1, \ell}} \ket*{r_{\Gen, 1, \ell}, r_{\Enc, 1, \ell}, r'_{\Enc, 1, \ell}} \]
        to register $\sfR_{\ell}$. In subregister $\sfE_{\ell}$, coherently compute 
            \[ \pk_{1, \ell} \gets \QHE.\Gen(1^{\secp}; r_{\Gen, 1, \ell}). \]
        If $\ell = \rnd$, coherently compute
            \[ \widehat{q_{\ell}}^{\ell} = \QHE.\Enc(\pk_{1, \ell}, q_{\ell}; r_{\Enc, 1, \ell}) \]
        in register $\sfE_{\ell}$; if $\ell > \rnd$, then instead compute 
            \[ \widehat{0}^{\ell} = \QHE.\Enc(\pk_{1, \ell}, 0; r_{\Enc, 1, \ell}). \]
        For $\max\{2, \rnd\} \leq \ell \leq m$, additionally coherently compute in subregister $\sfE_{\ell}$ the value
            \[ \widehat{\sk}_{1, \ell-1} = \QHE.\Enc(\pk_{1, \ell}, \sk_{1, \ell-1}; r'_{\Enc, 1, \ell}). \]
        \item Apply the circuit $(C^{\ketbra*{\Enc(q_{1, \rnd})^{\rnd}}, \Pi_{\Acc, \rnd}}_{\eps, \delta})^{\dagger}$ with ancilla register $\sfO$.
        \item Measure the projection $\Pi_{\Acc, \rnd}$; if the outcome is 0 then abort/output $\bot$.
        
        Otherwise, apply $U_{\Ver, \rnd}$ and measure registers $\sfR_{\rnd}$, $\sfE_{\rnd}$, and $\sfZ_{\rnd}$ to get measurement outcomes $(r_{\Gen, 1, \rnd}, r_{\Enc, 1, \rnd}, r'_{\Enc, 1, \rnd})$, $(\pk_{1, \rnd}, \widehat{\sk}_{1, \rnd-1}, \widehat{q}_{1, \rnd})$, and $\widehat{a}_{[k], \rnd}$.

        Set $\tau_{\rnd} = \tau_{\rnd-1} || \pk_{[k], \rnd} || \widehat{\sk}_{[k], \rnd-1} || \widehat{q}_{[k], \rnd} || \widehat{a}_{[k], \rnd}$. If $\rnd = m$, skip to step 4.
        \item Apply the circuit $C^{\ketbra*{\Enc(0)^{\rnd+1}}, \Pi_{\Acc, \rnd+1}}_{\eps, \delta}$ with ancilla register $\sfO$.
        \item Measure the projection $\ketbra*{\Enc(0)^{\rnd+1}}$ before discarding registers $\sfR_{\rnd+1} \ldots \sfR_m$ and $\sfE_{\rnd+1} \ldots \sfE_m$.
        \item If the outcome of the projection is 0, discard registers $\sfO\sfZ_{\rnd+1}\ldots\sfZ_m\sfI$ and abort/output $\bot$.
        \item Compute $a_{1, \rnd} = \QHE.\Dec(\QHE.\Gen(1^{\secp}; r_{\Gen, 1, \rnd}), \widehat{a}_{1, \rnd})$ before sending $a_{1, \rnd}$ to $V$.
    \end{enumerate}
    \item Output $\tau_m$.
\end{enumerate}
}
\noindent
\textbf{Runtime.}
The fact that $\calA$ runs in time $\poly(\secp)$ follows by the efficiency of the QSVT circuits and the fact that $\log 1/\eps, 1/\nu, m, k, |\calB| = \poly(\secp)$.
$ $\newline
\textbf{Soundness.}
\noindent
Since our reduction only calls $\calB$ on the specified security parameter $\secp$, we assume for simplicity that $\Lambda = \bbN$. We also assume for the sake of notational simplicity that $\secp$ is sufficiently large so that $2^{-\secp} < \nu/10k$. To show that $\calA$ succeeds with the claimed probability, we consider the following set of hybrids:
\begin{itemize}
    \item Hybrid $\calH_0$: Run a random execution of $\calA$.
    \item Hybrid $\calH_{1, 1}$: Like $\calH_0$, except in steps 2(b)-(f), use the external verifier's randomness to compute $\ket*{\Enc(q_1)^1}$ instead of $\ket*{\Enc(0)^1}$.
    \item Hybrid $\calH_{1, 2}$: Like $\calH_{1, 1}$, but in step 2(e), do not discard any registers. If the projection rejects in step 2(e), then additionally discard $\sfR\sfE$ in step 2(f). If the projection accepts, then skip step 3(b) for $\rnd = 1$.
    \item Hybrid $\calH_{1, 3}$: Like $\calH_{1, 2}$, but also dropping steps 2(e)-(f).
    \item Hybrid $\calH_{1, 4}$: Like $\calH_{1, 3}$, but also dropping step 2(d) and iteration $s = 1$ of step 3(c).
\end{itemize}
We now bound the trace distance between each pair of hybrids.

At a high level, the only difference between $\calH_0$ and $\calH_{1, 1}$ is that in $\calH_0$ the coherent ciphertext state in coordinate $i$ encrypts the dummy message $0$ whereas in $\calH_{1, 1}$ it encrypts the actual query $q_1$. Although semantic security alone gives only computational indistinguishability, Lemma~\ref{lemma:comp_to_stat} shows that the entire output of the QSVT-and-return procedure is statistically close in the two cases, since the procedure never accesses the purifying randomness directly and retains the residual state only after successfully returning the ciphertext registers to the supplied encryption state.

We now give the formal proof of Lemma \ref{lemma:hybrids_poly_swap_2bf}.
\begin{lemma}\label{lemma:hybrids_poly_swap_2bf}
    $\TD(\calH_0, \calH_{1, 1}) = \negl(\secp)$.
\end{lemma}
\begin{proof}[Proof of Lemma \ref{lemma:hybrids_poly_swap_2bf}]
We will couple the two hybrids based on the external randomness $r_{\Query, 1}$ sampled by the external verifier.

As $\calH_0$ and $\calH_{1, 1}$ are identical from step 3 onwards, it suffices to argue that the states of the hybrids at the end of step 2 are close in trace distance. We observe that $\calA$ either (1) aborts in step 2 after $N$ iterations or (2) it succeeds in some iteration $s = \iter$ with sampled randomness $r_{\Ver}$ and pure state $\ket{0}_{\sfO} \otimes \ket*{\psi_{r_{\Ver}}}$ (fully determined by $r_{\Ver}$) left in registers $\sfO\sfZ_1\ldots\sfZ_m\sfI$ (the fact that the $\sfO$ register is unentangled and left in the zero state follows from Theorem \ref{thm:qsvt}). In the first case we say the state of $\calA$ at the end of step 2 is $(\bot, \bot)$, and in the second case we label the state by $(r_{\Ver}, \ket*{\psi_{r_{\Ver}}})$. Bounding $\TD(\calH_0, \calH_{1, 1})$ follows from bounding the resulting distributions over states at the end of step 2 as defined.

We observe that this state can equivalently be calculated in parallel rather than sequentially: if we run $N$ independent executions of step 2, there is a deterministic function which maps the $N$ outcomes to the resulting state we want. By the data processing inequality and a union bound, we see that it suffices to argue that a \emph{single} execution of step 2 in $\calH_0$ and $\calH_{1, 1}$ results in negligibly close states, since $N = 20k\secp/\nu = \poly(\secp)$.

We now couple the choice of $r_{\Ver}$ between the hybrids. Note that in both hybrids, the probability of sampling $r_{\Ver}$ is identical. We can therefore fix $r_{\Ver} = r^{*}_{\Ver}$.

We now define two distributions:
\begin{itemize}
    \item $\calD_{0, \secp}$ is sampled as follows: first, sample randomness $r_{\Gen, 1}$, $r_{\Enc, 1}$, and $r'_{\Enc, 1}$. Compute 
        \[ (\pk_{1, \ell}, \sk_{1, \ell}) = \QHE.\Gen(1^{\secp}; r_{\Gen, 1, \ell}) \]
    for $1 \leq \ell \leq m$ and 
        \[ \widehat{\sk}_{1, \ell-1} = \QHE.\Enc(\pk_{1, \ell}, \sk_{1, \ell-1}; r'_{\Enc, 1, \ell}) \]
    and 
        \[ \ct_{\ell} = \QHE.\Enc(\pk_{1, \ell}, 0; r_{\Enc, 1, \ell}) \]
    for $2 \leq \ell \leq m$. Finally, compute $\ct_1 = \QHE.\Enc(\pk_{1, 1}, 0; r_{\Enc, 1, 1})$ before outputting $(\pk_{1, 1}, \ct_1, \pk_{1, 2}, \widehat{\sk}_{1, 1}, \ct_2, \ldots, \pk_{1, m}, \widehat{\sk}_{1, m-1}, \ct_m)$.
    \item $\calD_{1, \secp}$ is sampled as follows: first, sample randomness $r_{\Gen, 1}$, $r_{\Enc, 1}$, and $r'_{\Enc, 1}$. Compute 
        \[ (\pk_{1, \ell}, \sk_{1, \ell}) = \QHE.\Gen(1^{\secp}; r_{\Gen, 1, \ell}) \]
    for $1 \leq \ell \leq m$ and 
        \[ \widehat{\sk}_{1, \ell-1} = \QHE.\Enc(\pk_{1, \ell}, \sk_{1, \ell-1}; r'_{\Enc, 1, \ell}) \]
    and 
        \[ \ct_{\ell} = \QHE.\Enc(\pk_{1, \ell}, 0; r_{\Enc, 1, \ell}) \]
    for $2 \leq \ell \leq m$. Finally, compute $\ct_1 = \QHE.\Enc(\pk_{1, 1}, q_1; r_{\Enc, 1, 1})$ before outputting $(\pk_{1, 1}, \ct_1, \pk_{1, 2}, \widehat{\sk}_{1, 1}, \ct_2, \ldots, \pk_{1, m}, \widehat{\sk}_{1, m-1}, \ct_m)$.
\end{itemize}

\begin{claim}\label{claim:comp_indist}
    $\calD_{0, \secp}$ and $\calD_{1, \secp}$ are computationally indistinguishable to $\poly(\secp)$-size quantum adversaries with quantum advice.
\end{claim}
\begin{proof}[Proof of Claim \ref{claim:comp_indist}]
    We will use a hybrid argument to prove our claim. For $0 \leq \rnd \leq m$ and $b \in \{0, 1\}$, we define the hybrid $\calD_{b, \rnd, \secp}$ as follows:
    \begin{enumerate}
        \item Sample randomness $r_{\Gen, 1}$,  $r_{\Enc, 1}$, and $r'_{\Enc, 1}$.
        \item Compute 
            \[ (\pk_{1, \ell}, \sk_{1, \ell}) = \QHE.\Gen(1^{\secp}; r_{\Gen, 1, \ell}) \]
        for $1 \leq \ell \leq m$.
        \item Compute
            \[ \ct_{\ell} = \QHE.\Enc(\pk_{1, \ell}, 0; r_{\Enc, 1, \ell}) \]
        for $2 \leq \ell \leq m$, and
            \[ \ct_1 = \QHE.\Enc(\pk_{1, 1}, 0; r_{\Enc, 1, 1}) \]
        if $b = 0$, or
            \[ \ct_1 = \QHE.\Enc(\pk_{1, 1}, q_1; r_{\Enc, 1, 1}) \]
        if $b = 1$.
        \item For $2 \leq \ell \leq \rnd$, compute
            \[ \widehat{\sk}_{1, \ell-1} = \QHE.\Enc(\pk_{1, \ell}, \sk_{1, \ell-1}; r'_{\Enc, 1, \ell}) \]
        and for $\ell > \rnd$, compute
            \[ \widehat{\sk}_{1, \ell-1} = \QHE.\Enc(\pk_{1, \ell}, 0; r'_{\Enc, 1, \ell}). \]
        \item Output $(\pk_{1, 1}, \ct_1, \pk_{1, 2}, \widehat{\sk}_{1, 1}, \ct_2, \ldots, \pk_{1, m}, \widehat{\sk}_{1, m-1}, \ct_m)$.
    \end{enumerate}
    Note that by construction, $\calD_{0, \secp} = \calD_{0, m, \secp}$ and $\calD_{1, \secp} = \calD_{1, m, \secp}$. We begin by arguing that $\calD_{b, \rnd, \secp}$ and $\calD_{b, \rnd+1, \secp}$ are computationally indistinguishable. 
    
    Suppose for the sake of contradiction that there exists a $\QPT$ distinguisher $\calA$, $\poly(\secp)$-size advice $\ket{\psi}$, and polynomial $p(\secp)$ such that for infinitely many $\secp \in \bbN$,
        \[ \left|\Pr_{x \gets \calD_{b, \rnd, \secp}}[\calA(x, \ket{\psi}) = 1] - \Pr_{x \gets \calD_{b, \rnd+1, \secp}}[\calA(x, \ket{\psi}) = 1]\right| \geq \frac{1}{p(\secp)}. \]
    Then, by an averaging argument, there exists 
        \[ \mathsf{adv}_{\Gen} = (r_{\Gen, 1, 1}, \ldots, r_{\Gen, 1, \rnd}, r_{\Gen, 1, \rnd+2}, \ldots, r_{\Gen, 1, m}) \]
    and 
        \[ \mathsf{adv}_{\Enc} = (r_{\Enc, 1, 1}, \ldots, r_{\Enc, 1, m}, r'_{\Enc, 1, 2}, \ldots, r'_{\Enc, 1, \rnd}, r'_{\Enc, 1, \rnd+2}, \ldots, r'_{\Enc, 1, m}) \]
    such that
        \[ \left|\Pr_{x \gets \calD_{b, \rnd, \secp}}[\calA(x, \ket{\psi}) = 1 \mid \mathsf{adv}_{\Gen}, \mathsf{adv}_{\Enc}] - \Pr_{x \gets \calD_{b, \rnd+1, \secp}}[\calA(x, \ket{\psi}) = 1 \mid \mathsf{adv}_{\Gen}, \mathsf{adv}_{\Enc}]\right| \geq \frac{1}{p(\secp)}. \]
    If we hardwire $(\mathsf{adv}_{\Gen}, \mathsf{adv}_{\Enc})$, then we see that there exists a $\QPT$ distinguisher $\calB$ such that for infinitely many $\secp \in \bbN$,
    \begin{align*}
        &\Biggl|
        \Pr_{\substack{
            (\pk_{\rnd+1}, \sk_{\rnd+1}) \gets \QHE.\Gen(1^{\secp}) \\
            \ct_0 \gets \QHE.\Enc(\pk_{\rnd+1}, 0)
        }}
        \!\left[
            \calB(\pk_{\rnd+1}, \ct_0, \ket{\psi},
            \mathsf{adv}_{\Gen}, \mathsf{adv}_{\Enc}) = 1
        \right] \\
        &\qquad -
        \Pr_{\substack{
            (\pk_{\rnd+1}, \sk_{\rnd+1}) \gets \QHE.\Gen(1^{\secp}) \\
            \ct_1 \gets \QHE.\Enc(\pk_{\rnd+1}, \sk_{\rnd})
        }}
        \!\left[
            \calB(\pk_{\rnd+1}, \ct_1, \ket{\psi},
            \mathsf{adv}_{\Gen}, \mathsf{adv}_{\Enc}) = 1
        \right]
        \Biggr|
        \geq \frac{1}{p(\secp)},
    \end{align*}
    contradicting the non-uniform security of $\QHE$. We conclude that $\calD_{b, \rnd, \secp} \stackrel{\negl}\approx \calD_{b, \rnd+1, \secp}$ for all $1 \leq \rnd \leq m-1$. A nearly identical argument then shows that $\calD_{0, 1, \secp} \stackrel{\negl}\approx \calD_{1, 1, \secp}$ (since the distinguisher no longer has access to any encryptions of $\sk_1$). Thus, we see that $\calD_{0, \secp} \stackrel{\negl}\approx \calD_{1, \secp}$ since $m = \poly(\secp)$, as desired.
\end{proof}
Since $\QHE.\Gen$ and $\QHE.\Enc$ are $\PPT$ algorithms (for classical messages), $\{D_{0, \secp}\}_{\secp \in \bbN}$ and $\{D_{1, \secp}\}_{\secp \in \bbN}$ are also classically efficiently sampleable distributions.

With this claim we are now almost done with the proof. We slightly tweak our reduction to an observationally equivalent form where the measurement in step 2(c) is deferred until after step 2(f). Concretely, this can be done by adding an ancilla qubit initialized to $\ket{0}$ in register $\sfA$, applying the unitary 
    \[ U_{\sfA\sfE\sfZ\sfI} = X_{\sfA} \otimes \Pi_{\Acc, 1, \sfE\sfZ\sfI} + I_{\sfA} \otimes (I-\Pi_{\Acc})_{\sfE\sfZ\sfI} \]
in place of measuring $\Pi_{\Acc}$, and measuring $\sfA$ after step 2(f). We note that in this new circuit, the measurement of register $\sfA$ commutes with (1) the QSVT circuit in step 2(d), (2) the measurement of the projection $\ketbra*{D_{b, \secp}}_{\sfR\sfE}$ (which is $\ketbra*{\Enc(0)^1}_{\sfR\sfE}$ in $\calH_0$ and $\ketbra*{\Enc(q_1)^1}$ in $\calH_{1, 1}$) and discarding of registers $\sfR\sfE$ in step 2(e), and (3) the possible discarding of registers $\sfO\sfZ\sfI$ in step 2(f). 

We now note that in steps 2(b)-(e) before the $\sfR$ and $\sfE$ registers are traced out, our new circuit consists of appending $\ket*{D_{b, \secp}}$ (which is $\ket*{\Enc(0)^1}_{\sfR\sfE}$ in $\calH_0$ and $\ket*{\Enc(q_1)^1}_{\sfR\sfE}$ in $\calH_{1, 1}$) to a pure state $\ket*{\psi_{r^{*}_{\Ver}}}$ before applying a series of two types of unitaries: (1) unitaries (such as the deferred implementation of $\Pi_{\Acc, 1}$) which do not act on register $\sfR$, or (2) calls to
    \[ C_{\ketbra*{D_{b, \secp}}}NOT_{\sfO\sfR\sfE} := X_{\sfO} \otimes \ketbra*{D_{b, \secp}}_{\sfR\sfE} + I_{\sfO} \otimes (I-\ketbra*{D_{b, \secp}})_{\sfR\sfE}, \]
each of which is simply a single call to the $U_{b, \secp}$ oracle! In addition, when the projection $\ketbra*{D_{b, \secp}}$ outputs 0, our reduction simply aborts, which is equivalent to replacing the entire state with $\ket{\bot}$. This is therefore precisely the setting in which Lemma \ref{lemma:comp_to_stat} applies, which tells us that the states of $\calH_0$ and $\calH_{1, 1}$ right after step 2(f) are statistically negligibly close. Since the measurement of register $\sfA$ is done in both hybrids after step 2(f), this cannot increase the trace distance. We therefore conclude that a single execution of step 2 in $\calH_0$ and $\calH_{1, 1}$ result in negligibly close states, concluding the proof of Lemma \ref{lemma:hybrids_poly_swap_2bf}.
\end{proof}

\begin{lemma}\label{lemma:hybrids_poly_first_rewrite}
    $\TD(\calH_{1, 1}, \calH_{1, 2}) = 0$.
\end{lemma}
\begin{proof}[Proof of Lemma \ref{lemma:hybrids_poly_first_rewrite}]
    Note that in $\calH_{1, 1}$, we only continue to step 3 if $\ketbra*{\Enc(q_1)^1}$ accepts. When this projection accepts, step 3(b) reinitializes $\sfR\sfE$ to the \emph{same} state $\ket*{\Enc(q_1)^1}$. We can therefore omit step 3(b) and defer the discarding of registers $\sfR\sfE$ to step 2(e), as desired.
\end{proof}

We now argue that $\calH_{1, 2}$ and $\calH_{1, 3}$ are close in trace distance. Intuitively, this follows from the fact that $\calH_{1, 2}$ and $\calH_{1, 3}$ differ only in whether, after applying the QSVT algorithm, we explicitly check that the encryption registers have been returned to the clean state $\ket{\Enc(q_1)}$, and abort if this check fails. The key point is that this check is almost redundant. Before projecting onto the favorable-continuation subspace $\Pi_\Acc$, the state lies exactly in the clean-encryption subspace. In the joint singular-value decomposition of $\Pi_{\Acc}$ and $\ket{\Enc(q_1)}\bra{\Enc(q_1)}$, projecting onto $\Pi_{\Acc}$ suppresses components having small overlap with the clean-encryption subspace. Thus, unless the $\Pi_{\Acc}$-projection itself succeeds with very small probability, in which case both hybrids almost always abort, the normalized post-selected state has negligible weight on such low-overlap components. The QSVT algorithm then rotates essentially all of the remaining state back into the clean-encryption subspace. Consequently, the check in $\calH_{1, 2}$ fails only with very small probability and, by gentle measurement, conditioning on its success barely disturbs the state.
\begin{lemma}\label{lemma:hybrids_poly_qsvt}
    $\TD(\calH_{1, 2}, \calH_{1, 3}) \leq \mu+\negl(\secp) = \frac{\nu}{10km}+\negl(\secp)$.
\end{lemma}
\begin{proof}[Proof of Lemma \ref{lemma:hybrids_poly_qsvt}]
We couple the two hybrids based on the external verifier's randomness $r_{\Query, 1}$ as well as the choice of $r_{\Ver}$. 

As in Lemma \ref{lemma:hybrids_poly_swap_2bf}, it suffices to argue that the trace distance between a single execution of step 2 in $\calH_{1, 2}$ and $\calH_{1, 3}$ is bounded by $\mu/N+\negl(\secp)$, since this results in a total trace distance of $(\frac{\mu}{N} + \negl(\secp)) \cdot N = \mu + \negl(\secp)$.

We argue that this trace distance can be bounded regardless of the choice of $r_{\Ver}$. Fix any such choice and consider the state $\ket*{\psi_{r_{\Ver}}}$ which remains at the end of step 2(b). If      \[ p_{\Acc, 1} := \|\Pi_{\Acc, 1} \ket*{\psi_{r_{\Ver}}} \otimes \ket*{\Enc(q_1)^1}\|^2 < \mu/N, \]
then both $\calH_{1, 2}$ and $\calH_{1, 3}$ abort/end early in step 2(c) with probability at least $1-\mu/N$, and thus are obviously $(\mu/N)$-close in trace distance.

We therefore assume that $p_{\Acc, 1} \geq \mu/N$ for the remainder of our analysis. We now consider the singular values (which we denote by $\varsigma_j$) of $\ket*{\psi_{r_{\Ver}}} \otimes \ket*{\Enc(q_1)^1}$ with respect to $\Pi_A = \Pi_{\Acc, 1}$ and $\Pi_B = \ketbra*{\Enc(q_1)^1}$.

First note that if we write 
    \[ \ket*{\psi_{r_{\Ver}}} \otimes \ket*{\Enc(q_1)^1} = \sum_j \alpha_j \ket*{w_j}, \]
we have that
    \[ \ket{\upsilon_1} := \frac{\Pi_{\Acc, 1} \ket*{\psi_{r_{\Ver}}} \otimes \ket*{\Enc(q_1)^1}}{\|\Pi_{\Acc, 1} \ket*{\psi_{r_{\Ver}}} \otimes \ket*{\Enc(q_1)^1}\|} := \sum_j \frac{\alpha_j\varsigma_j}{\sqrt{p_{\Acc, 1}}} \ket*{v_j}. \]
Subsequently, we know that
\begin{align*}
    \sum_{j: \varsigma_j < \delta} \frac{|\alpha_j|^2\varsigma_j^2}{p_{\Acc, 1}} \leq \sum_{j: \varsigma_j < \delta} \frac{|\alpha_j|^2 \delta^2}{p_{\Acc, 1}} \leq \frac{\delta^2}{p_{\Acc, 1}} \leq \frac{\mu^4/4N^4}{\mu/N} = \frac{\mu^3}{4N^3} \leq \frac{\mu^2}{4N^2}.
\end{align*}
By Theorem \ref{thm:qsvt}, we know that $C^{\ketbra*{\Enc(q_1)^1}, \Pi_{\Acc, 1}}_{\eps, \delta} \ket{\upsilon_1} \ket{0} = \ket{\sigma_1} \ket{0}$ for some pure state $\ket{\sigma_1}$.

We now define the probability
    \[ p_{q_1} := \|(\ketbra*{\Enc(q_1)^1} \otimes I_{\sfO}) \ket{\sigma_1} \ket{0}_{\sfO}\|^2_2 \]
and state
    \[ \ket*{\sigma'_1} \otimes \ket{0} := \frac{(\ketbra*{\Enc(q_1)^1} \otimes I) \ket{\sigma_1} \ket{0}}{\|(\ketbra*{\Enc(q_1)^1} \otimes I) \ket{\sigma_1} \ket{0}\|}. \]
As $\ket{\upsilon_1} \in \im(\Pi_{\Acc, 1})$, we can apply Corollary \ref{cor:qsvt_threshold} to see that
    \[ p_{q_1} \geq 1-\frac{\mu^2}{4N^2}-\negl(\secp). \]
By the Gentle Measurement Lemma, we know that
    \[ \TD\left(\ket*{\sigma'_1} \ket{0}, \ket{\sigma_1} \ket{0}\right) \leq \frac{\mu}{2N}+\negl(\secp). \]
Our trace distance is therefore at most
\begin{align*}
    &\TD\left((1-p_{q_1}) \cdot \ketbra{0} \otimes \ketbra{\bot} + p_{q_1} \cdot \ketbra{1} \otimes \ketbra*{\sigma'_1, 0}, \ketbra{1} \otimes \ketbra*{\sigma_1, 0}\right) \\
    \leq \,\, &2(1-p_{q_1}) + p_{q_1} \cdot \TD(\ketbra*{\sigma'_1, 0}, \ketbra*{\sigma_1, 0}) \\
    \leq \,\, &2(1-p_{q_1}) + p_{q_1} \cdot \left[\frac{\mu}{2N} + \negl(\secp)\right] \\
    \leq \,\, &2\left[\frac{\mu^2}{4N^2}+\negl(\secp)\right] + \left[\frac{\mu}{2N}+\negl(\secp)\right] \\
    \leq \,\, &\frac{\mu}{N}+\negl(\secp),
\end{align*}
as desired.
\end{proof}

\begin{lemma}\label{lemma:hybrids_poly_second_rewrite}
    $\TD(\calH_{1, 3}, \calH_{1, 4}) = 0$.
\end{lemma}
\begin{proof}[Proof of Lemma \ref{lemma:hybrids_poly_second_rewrite}]
    This follows directly from the fact that steps 2(d) and 3(c) are conjugates of each other and are applied in succession (step 3(a) is classical and has no effect on the quantum state of the reduction).
\end{proof}

At this stage, hybrid $\calH_{1, 4}$ looks as follows:
\algo{Hybrid $\calH_{1, 4}$}{poly-hybrid-four}{
$\calA$ begins with state $\ket{\psi}^{\otimes N}$, and runs as follows:
\begin{enumerate}
    \item Set $\eps = 2^{-\secp}$, $\mu = \nu/10km$, $\delta =  \mu^2/2N^2$, $\tau_0 := \emptyset$ and $\widehat{\sk}_{[k], 0} := \emptyset$. Sample $r_{\Query, 1}$.
    \item For $s = 1$ to $N$:
    \begin{enumerate}
        \item Sample the following sets of randomness: 
        \begin{align*}
            \{r_{\Gen, j, \ell}, r_{\Enc, j, \ell}\}_{j \neq 1, \ell \in [m]}, \{r'_{\Enc, j, \ell}\}_{j \neq 1, 2 \leq \ell \leq m}, \{r_{\Query, j}\}_{j \neq 1}.
        \end{align*}
        \item Append the state $\ket*{\Enc(q_1)^1}_{\sfR\sfE}$, where $q_{1, 1} = q_1$ is derived from $r_{\Query, 1}$.
        \item Measure $\Pi_{\Acc, 1}$; if the outcome is 0, return to the start of step 2 if $s \neq N$ and abort/output $\bot$ if $s = N$.
        \item Measure the projection $\Pi_{\Acc, 1}$; if the outcome is 0 then abort/output $\bot$.
        
        Otherwise, apply $U_{\Ver, 1}$ and measure registers $\sfR_1$, $\sfE_1$, and $\sfZ_1$ to get measurement outcomes $(r_{\Gen, 1, 1}, r_{\Enc, 1, 1})$, $(\pk_{1, 1}, \widehat{q}_{1, 1})$, and $\widehat{a}_{[k], 1}$.
        
        Set $\tau_1 = \tau_0 || \pk_{[k], 1} || \widehat{q}_{[k], 1} || \widehat{a}_{[k], 1}$.
        \item Apply the circuit $C^{\ketbra*{\Enc(0)^2}, \Pi_{\Acc, 2}}_{\eps, \delta}$ with ancilla qubit $\ket{0}_{\sfO}$.
        \item Measure the projection $\ketbra*{\Enc(0)^2}$ before discarding registers $\sfR_2 \ldots \sfR_m$ and $\sfE_2 \ldots \sfE_m$.
        \item If the outcome of the projection is 0, discard registers $\sfO\sfZ_2\ldots\sfZ_m\sfI$ and abort/output $\bot$. Otherwise, continue to step 3.
    \end{enumerate}
    \item For $\rnd = 2$ to $m$:
    \begin{enumerate}
        \item Compute the $(\rnd)$-th query $q_{1, \rnd} = q_{\rnd}$ using $r_{\Query, 1}$, the secret keys of rounds $1$ through $\rnd-1$, and $\tau_{\rnd-1}$.
        \item Append the state $\ket*{\Enc(q_{1, \rnd})^{\rnd}}_{\sfR_{\rnd}\ldots\sfR_m\sfE_{\rnd}\ldots\sfE_m}$.
        \item Apply the circuit $(C^{\ketbra*{\Enc(q_{1, \rnd})^{\rnd}}, \Pi_{\Acc, \rnd}}_{\eps, \delta})^{\dagger}$ with ancilla register $\sfO$.
        \item Measure the projection $\Pi_{\Acc, \rnd}$; if the outcome is 0 then abort/output $\bot$.
        
        Otherwise, apply $U_{\Ver, \rnd}$ and measure registers $\sfR_{\rnd}$, $\sfE_{\rnd}$, and $\sfZ_{\rnd}$ to get measurement outcomes $(r_{\Gen, 1, \rnd}, r_{\Enc, 1, \rnd}, r'_{\Enc, 1, \rnd})$, $(\pk_{1, \rnd}, \widehat{\sk}_{1, \rnd-1}, \widehat{q}_{1, \rnd})$, and $\widehat{a}_{[k], \rnd}$.

        Set $\tau_{\rnd} = \tau_{\rnd-1} || \pk_{[k], \rnd} || \widehat{\sk}_{[k], \rnd-1} || \widehat{q}_{[k], \rnd} || \widehat{a}_{[k], \rnd}$. If $\rnd = m$, skip to step 4.
        \item Apply the circuit $C^{\ketbra*{\Enc(0)^{\rnd+1}}, \Pi_{\Acc, \rnd+1}}_{\eps, \delta}$ with ancilla register $\sfO$.
        \item Measure the projection $\ketbra*{\Enc(0)^{\rnd+1}}$ before discarding registers $\sfR_{\rnd+1} \ldots \sfR_m$ and $\sfE_{\rnd+1} \ldots \sfE_m$.
        \item If the outcome of the projection is 0, discard registers $\sfO\sfZ_{\rnd+1}\ldots\sfZ_m\sfI$ and abort/output $\bot$.
    \end{enumerate}
    \item Output $\tau_m$.
\end{enumerate}
}
We observe that if step 2(c) succeeds, then the measurement in step 2(d) must output 1, so we can remove this measurement. The application of $U_{\Ver, 1}$ in step 2(d) undoes the outer $U_{\Ver, 1}^{\dagger}$ in the projection $\Pi_{\Acc, 1}$ of step 2(c), so that the net effect in step 2(c) is to first apply $U_{\Ver, 1}$ before measuring the projection $\Pi_{\Acc, 2}$. But since in step 2(d) we measure registers $\sfR_1\sfE_1\sfZ_1$ immediately after and $\Pi_{\Acc, 2}$ is controlled on (or does not touch) these registers, we can commute their measurements. Thus, $\calH_{1, 4}$ is equivalent to the following hybrid $\calH_{1, 4'}$ (with the modified steps highlighted in blue):
\algo{Hybrid $\calH_{1, 4'}$}{poly-hybrid-four-prime}{
$\calA$ begins with state $\ket{\psi}^{\otimes N}$, and runs as follows:
\begin{enumerate}
    \item Set $\eps = 2^{-\secp}$, $\mu = \nu/10km$, $\delta =  \mu^2/2N^2$, $\tau_0 := \emptyset$ and $\widehat{\sk}_{[k], 0} := \emptyset$. Sample $r_{\Query, 1}$.
    \item For $s = 1$ to $N$:
    \begin{enumerate}
        \item Sample the following sets of randomness: 
        \begin{align*}
            \{r_{\Gen, j, \ell}, r_{\Enc, j, \ell}\}_{j \neq 1, \ell \in [m]}, \{r'_{\Enc, j, \ell}\}_{j \neq 1, 2 \leq \ell \leq m}, \{r_{\Query, j}\}_{j \neq 1}.
        \end{align*}
        \item Append the state $\ket*{\Enc(q_1)^1}_{\sfR\sfE}$, where $q_{1, 1} = q_1$ is derived from $r_{\Query, 1}$. 
        \item \textbf{\color{blue} Apply $U_{\Ver, 1}$ and measure registers $\sfR_1$, $\sfE_1$, and $\sfZ_1$ to get measurement outcomes $(r_{\Gen, 1, 1}, r_{\Enc, 1, 1})$, $(\pk_{1, 1}, \widehat{q}_{1, 1})$, and $\widehat{a}_{[k], 1}$. Set $\tau_1 = \tau_0 || \pk_{[k], 1} || \widehat{q}_{[k], 1} || \widehat{a}_{[k], 1}$.}
        \item \textbf{\color{blue} Measure $\Pi_{\Acc, 2}$; if the outcome is 0, return to the start of step 2 if $s \neq N$ and abort/output $\bot$ if $s = N$.}
        \item Apply the circuit $C^{\ketbra*{\Enc(0)^2}, \Pi_{\Acc, 2}}_{\eps, \delta}$ with ancilla register $\sfO$.
        \item Measure the projection $\ketbra*{\Enc(0)^2}$ before discarding registers $\sfR_2 \ldots \sfR_m$ and $\sfE_2 \ldots \sfE_m$.
        \item If the outcome of the projection is 0, discard registers $\sfO\sfZ_2\ldots\sfZ_m\sfI$ and abort/output $\bot$. Otherwise, continue to step 3.
    \end{enumerate}
    \item For $\rnd = 2$ to $m$:
    \begin{enumerate}
        \item Compute the $(\rnd)$-th query $q_{1, \rnd} = q_{\rnd}$ using $r_{\Query, 1}$, the secret keys of rounds $1$ through $\rnd-1$, and $\tau_{\rnd-1}$.
        \item Append the state $\ket*{\Enc(q_{1, \rnd})^{\rnd}}_{\sfR_{\rnd}\ldots\sfR_m\sfE_{\rnd}\ldots\sfE_m}$.
        \item Apply the circuit $(C^{\ketbra*{\Enc(q_{1, \rnd})^{\rnd}}, \Pi_{\Acc, \rnd}}_{\eps, \delta})^{\dagger}$ with ancilla register $\sfO$.
        \item Measure the projection $\Pi_{\Acc, \rnd}$; if the outcome is 0 then abort/output $\bot$.
        
        Otherwise, apply $U_{\Ver, \rnd}$ and measure registers $\sfR_{\rnd}$, $\sfE_{\rnd}$, and $\sfZ_{\rnd}$ to get measurement outcomes $(r_{\Gen, 1, \rnd}, r_{\Enc, 1, \rnd}, r'_{\Enc, 1, \rnd})$, $(\pk_{1, \rnd}, \widehat{\sk}_{1, \rnd-1}, \widehat{q}_{1, \rnd})$, and $\widehat{a}_{[k], \rnd}$.

        Set $\tau_{\rnd} = \tau_{\rnd-1} || \pk_{[k], \rnd} || \widehat{\sk}_{[k], \rnd-1} || \widehat{q}_{[k], \rnd} || \widehat{a}_{[k], \rnd}$. If $\rnd = m$, skip to step 4.
        \item Apply the circuit $C^{\ketbra*{\Enc(0)^{\rnd+1}}, \Pi_{\Acc, \rnd+1}}_{\eps, \delta}$ with ancilla register $\sfO$.
        \item Measure the projection $\ketbra*{\Enc(0)^{\rnd+1}}$ before discarding registers $\sfR_{\rnd+1} \ldots \sfR_m$ and $\sfE_{\rnd+1} \ldots \sfE_m$.
        \item If the outcome of the projection is 0, discard registers $\sfO\sfZ_{\rnd+1}\ldots\sfZ_m\sfI$ and abort/output $\bot$.
    \end{enumerate}
    \item Output $\tau_m$.
\end{enumerate}
}
Now note that $U_{\Ver, 1}$ is controlled on register $\sfE_1$ and applies to registers $\sfZ_1$ and $\sfI$, so we can also commute the measurements of registers $\sfR_1$ and $\sfE_1$ to immediately after step 1(b). But doing so means we immediately measure registers $\sfR_1\sfE_1$ of $\ket*{\Enc(q_1)^1}_{\sfR\sfE}$, which has the net effect of simply sampling random $(r_{\Gen, 1, 1}, r_{\Enc, 1, 1})$ and computing $(\pk_{1, 1}, \sk_{1, 1}) = \QHE.\Gen(1^{\secp}; r_{\Gen, 1, 1})$ and $\widehat{q_1}^1 = \QHE.\Enc(\pk_{1, 1}, q_{1, 1}; r_{\Enc, 1, 1})$ before applying $U_1(\pk_{[k], 1}, \widehat{q}_{[k], 1})$ to registers $\sfZ_1\sfI$. In total, $\calH_{1, 4'}$ is therefore equivalent to the following hybrid $\calH_1$ (with the modified steps highlighted in green):
\algo{Hybrid $\calH_1$}{poly-hybrid-one}{
$\calA$ begins with state $\ket{\psi}^{\otimes N}$, and runs as follows:
\begin{enumerate}
    \item Set $\eps = 2^{-\secp}$, $\mu = \nu/10km$, $\delta =  \mu^2/2N^2$, $\tau_0 := \emptyset$ and $\widehat{\sk}_{[k], 0} := \emptyset$. Sample $r_{\Query, 1}$.
    \item For $s = 1$ to $N$:
    \begin{enumerate}
        \item Sample the following sets of randomness: 
        \begin{align*}
            \{r_{\Gen, j, \ell}, r_{\Enc, j, \ell}\}_{j \neq 1, \ell \in [m]}, \{r'_{\Enc, j, \ell}\}_{j \neq 1, 2 \leq \ell \leq m}, \{r_{\Query, j}\}_{j \neq 1}.
        \end{align*}
        \item \textbf{\color{ForestGreen} Sample $r_{\Gen, 1, 1}$ and $r_{\Enc, 1, 1}$ before computing $(\pk_{1, 1}, \sk_{1, 1}) = \QHE.\Gen(1^{\secp}; r_{\Gen, 1, 1})$ and $\widehat{q_1}^1 = \QHE.\Enc(\pk_{1, 1}, q_{1, 1}; r_{\Enc, 1, 1})$.}
        \item \textbf{\color{ForestGreen} Apply $U_1(\pk_{[k], 1}, \widehat{q}_{[k], 1})$ on registers $\sfZ_1\sfI$ and measure register $\sfZ_1$ to get measurement outcome $\widehat{a}_{[k], 1}$. Set $\tau_1 = \tau_0 || \pk_{[k], 1} || \widehat{q}_{[k], 1} || \widehat{a}_{[k], 1}$.}
        \item \textbf{\color{ForestGreen} Append the state $\ket*{\Enc(0)^2}_{\sfR_2\ldots\sfR_m\sfE_2\ldots\sfE_m}$.}
        \item Measure $\Pi_{\Acc, 2}$; if the outcome is 0, return to the start of step 2 if $s \neq N$ and abort/output $\bot$ if $s = N$.
        \item Apply the circuit $C^{\ketbra*{\Enc(0)^2}, \Pi_{\Acc, 2}}_{\eps, \delta}$ with ancilla register $\sfO$.
        \item Measure the projection $\ketbra*{\Enc(0)^2}$ before discarding registers $\sfR_2 \ldots \sfR_m$ and $\sfE_2 \ldots \sfE_m$.
        \item If the outcome of the projection is 0, discard registers $\sfO\sfZ_2\ldots\sfZ_m\sfI$ and abort/output $\bot$. Otherwise, continue to step 3.
    \end{enumerate}
    \item For $\rnd = 2$ to $m$:
    \begin{enumerate}
        \item Compute the $(\rnd)$-th query $q_{1, \rnd} = q_{\rnd}$ using $r_{\Query, 1}$, the secret keys of rounds $1$ through $\rnd-1$, and $\tau_{\rnd-1}$.
        \item Append the state $\ket*{\Enc(q_{1, \rnd})^{\rnd}}_{\sfR_{\rnd}\ldots\sfR_m\sfE_{\rnd}\ldots\sfE_m}$.
        \item Apply the circuit $(C^{\ketbra*{\Enc(q_{1, \rnd})^{\rnd}}, \Pi_{\Acc, \rnd}}_{\eps, \delta})^{\dagger}$ with ancilla register $\sfO$.
        \item Measure the projection $\Pi_{\Acc, \rnd}$; if the outcome is 0 then abort/output $\bot$.
        
        Otherwise, apply $U_{\Ver, \rnd}$ and measure registers $\sfR_{\rnd}$, $\sfE_{\rnd}$, and $\sfZ_{\rnd}$ to get measurement outcomes $(r_{\Gen, 1, \rnd}, r_{\Enc, 1, \rnd}, r'_{\Enc, 1, \rnd})$, $(\pk_{1, \rnd}, \widehat{\sk}_{1, \rnd-1}, \widehat{q}_{1, \rnd})$, and $\widehat{a}_{[k], \rnd}$.

        Set $\tau_{\rnd} = \tau_{\rnd-1} || \pk_{[k], \rnd} || \widehat{\sk}_{[k], \rnd-1} || \widehat{q}_{[k], \rnd} || \widehat{a}_{[k], \rnd}$. If $\rnd = m$, skip to step 4.
        \item Apply the circuit $C^{\ketbra*{\Enc(0)^{\rnd+1}}, \Pi_{\Acc, \rnd+1}}_{\eps, \delta}$ with ancilla register $\sfO$.
        \item Measure the projection $\ketbra*{\Enc(0)^{\rnd+1}}$ before discarding registers $\sfR_{\rnd+1} \ldots \sfR_m$ and $\sfE_{\rnd+1} \ldots \sfE_m$.
        \item If the outcome of the projection is 0, discard registers $\sfO\sfZ_{\rnd+1}\ldots\sfZ_m\sfI$ and abort/output $\bot$.
    \end{enumerate}
    \item Output $\tau_m$.
\end{enumerate}
}
We conclude that by Lemmas \ref{lemma:hybrids_poly_swap_2bf}, \ref{lemma:hybrids_poly_first_rewrite}, \ref{lemma:hybrids_poly_qsvt}, and
\ref{lemma:hybrids_poly_second_rewrite}, 
\begin{align*}
    \TD(\calH_0, \calH_1) &\leq \TD(\calH_0, \calH_{1, 1}) + \sum_{i = 1}^3 \TD(\calH_{1, i}, \calH_{1, i+1}) + \TD(\calH_{1, 4}, \calH_1) \\
    &\leq \negl(\secp) + 0 + \frac{\nu}{10km} + \negl(\secp) + 0 + 0 \\
    &\leq \frac{\nu}{10km}+\negl(\secp).
\end{align*}
But note that we can now run a nearly identical argument for the second round of queries and responses, and so on! Repeating this argument until $\rnd = m$ gives us the hybrid $\calH_m$, where 
    \[ \TD(\calH_0, \calH_m) \leq \sum_{i = 1}^{m} \TD(\calH_{i-1}, \calH_i) \leq \frac{\nu}{10k}+\negl(\secp), \]
since $m = \poly(\secp)$:
\algo{Hybrid $\calH_m$}{poly-hybrid-last}{
$\calA$ begins with state $\ket{\psi}^{\otimes N}$, and runs as follows:
\begin{enumerate}
    \item Set $\eps = 2^{-\secp}$, $\mu = \nu/10km$, $\delta =  \mu^2/2N^2$, $\tau_0 := \emptyset$ and $\widehat{\sk}_{[k], 0} := \emptyset$. Sample $r_{\Query, 1}$.
    \item For $s = 1$ to $N$:
    \begin{enumerate}
        \item Sample the following sets of randomness: 
        \begin{align*}
            \{r_{\Gen, j, \ell}, r_{\Enc, j, \ell}\}_{j, \ell \in [m]}, \{r'_{\Enc, j, \ell}\}_{j, 2 \leq \ell \leq m}, \{r_{\Query, j}\}_{j \neq 1}.
        \end{align*}
        \item For $\rnd = 1$ to $m$:
        \begin{enumerate}
            \item Apply $U_\rnd$ on registers $\sfZ_\rnd\sfI$ and measure register $\sfZ_\rnd$ to get measurement outcome $\widehat{a}_{[k], \rnd}$. Set $\tau_\rnd = \tau_{\rnd-1} || \pk_{[k], \rnd} || \widehat{\sk}_{[k], \rnd-1} || \widehat{q}_{[k], \rnd} || \widehat{a}_{[k], \rnd}$.
        \end{enumerate}
        \item If $\sum_{i \neq 1} V((\tau_m)_i, r_{\Ver, i}) \neq t'-1$, return to the start of step 2 if $s \neq N$ and abort/output $\bot$ if $s = N$.
    \end{enumerate}
    \item Output $\tau_m$.
\end{enumerate}
}
\noindent
Thus, $\calH_m$ can be described as follows:
\begin{enumerate}
    \item Sample $r_{\Query, 1}$.
    \item For $s = 1$ to $N$, run a random execution of $\calB$ interacting with $\widehat{V}^{k, t}$ (with $r_{\Query, 1}$ hardwired as the randomness of the first subverifier underneath the $\QHE$). Take the first transcript $\tau$ (if it exists) such that $\sum_{i \neq 1} V(\tau_i, r_{\Ver, i}) = t'-1$ and output $\tau$. If no transcript satisfies this condition after $s = N$, abort and output $\bot$.
\end{enumerate}
It therefore follows that
\begin{align*}
    \Pr[\langle \calA, V \rangle(1^\secp) = 1] &= \Pr_{\tau \gets \calH_0}[\Accept_1(\tau) = 1] \\
    &\geq \Pr_{\tau \gets \calH_m}[\Accept_1(\tau) = 1] - \sum_{j = 1}^{m} \TD(\calH_{j-1}, \calH_j) \\
    &\geq \Pr_{\tau \gets \calH_m}[\Accept_1(\tau) = 1] - \frac{\nu}{10k} - \negl(\secp).
\end{align*}

It thus remains to give a lower bound on the probability that the reduction in $\calH_m$ succeeds.
\begin{lemma}\label{lemma:poly_tight_final_prob}
    $\Pr_{\tau \gets \calH_m}[\Accept_1(\tau) = 1] \geq \eps+\frac{3\nu}{5k}$.
\end{lemma}
The proof of Lemma \ref{lemma:poly_tight_final_prob} closely follows \cite{HS11,Chu11} so we defer the full proof to the appendix. We can now conclude that
\begin{align*}
    \Pr[\langle \calA, V \rangle (1^{\secp}) = 1] &\geq \Pr_{\tau \gets \calH_m}[\Accept_1(\tau) = 1] - \frac{\nu}{10k} - \negl(\secp) \\
    &\geq \eps + \frac{3\nu}{5k} - \frac{\nu}{10k} - \negl(\secp) = \eps + \frac{\nu}{2k} - \negl(\secp),
\end{align*}
as desired.
\end{proof}

\section*{AI Disclosure}
The first version of this manuscript, which contained all of the main ideas of the proof, was obtained entirely without the assistance of models. The proof was inspired by the correlation reduction strategy of \cite{CL10} for simulatable arguments.

ChatGPT 6 Astra then identified a minor but non-standard modification to the reduction which could be used to obtain tighter soundness. Upon internal discussion, the authors realized that the same improvement could alternatively be attained by using an appropriately modified version of the correlation reduction strategy previously used in the context of classical \emph{three-message} arguments \cite{HS11,Chu11}. These changes materially affected only the proofs of Lemmas \ref{lemma:corr_reduction} and \ref{lemma:poly_tight_final_prob} as reflected in the appendix.

\section{Acknowledgments}\label{sec:acknowledgments}
This work was done in part while AH, YK, and NS were visiting the Simons Institute for the Theory of Computing. Special thanks to Luowen Qian for very insightful discussions in the early stages of this project. ZB is supported by the Horizon Europe Research and Innovation Program via ERC Project ACQUA (Grant 101087742). YK and AH were supported by the U.S. National Science Foundation under award No. 2534400.

\bibliographystyle{alpha}
\bibliography{references.bib}

\appendix
\section{Missing Proofs}\label{app:missing_proofs}
\subsection{Proof of Lemma \ref{lemma:corr_reduction}}
\begin{proof}[Proof of Lemma \ref{lemma:corr_reduction}]
The transformation $\calT$ will determine a subset $S \subset [k]$ of coordinates, threshold $t'$, and coins $r^{*}_{\Query, j}$ for $j \notin S$. Upon doing so, it outputs $k' = |S|$ and the following adversary $\calA_S^{\calB}$: to interact with $\widehat{V}^{k', t'}$, $\calA_S$ simulates the interaction of $\langle \calB, \widehat{V}^{k, t} \rangle$ by internally simulating the verifiers for all coordinates outside of $S$ (fixing their query randomness to $r^{*}_{\Query}$) and having the external verifier $\widehat{V}^{k', t'}$ play the role of the coordinates in $S$. Clearly, $\calA_S$ runs in time $\poly(\secp, k)$ given oracle access to $\calB$ since the verifier and $\QHE$ are efficient.

We now describe how $\calT$ selects $S$, $t'$, and $r^{*}_{\Query}$:
\begin{enumerate}
    \item Initialize $S = [k]$ and $t' = t$. 
    \item While $|S| \neq 1$:
    \begin{enumerate}
        \item Let $i$ be the first coordinate in $S$.
        \item For $\iter = 1$ to $M_1 := 20k^2\secp/\nu$:
        \begin{enumerate}
            \item Sample $r_i$. Simulate the interaction $\langle \calB, \widehat{V}^{k, t}\rangle (1^{\secp})$ for $M_2 := 1600k^5\secp^2/\nu^3$ times, fixing query randomness $r_{\Query, j} = r^{*}_j$ for $j \notin S$ and $r_{\Query, i} = r_i$, using $M_2$ separate copies of the initial state of $\calB$.
            
            Using the $M_2$ interactions/transcripts, compute estimates $\widehat{p}_{t'}(r_i)$ (and $\widehat{p}_{t'-1}(r_i)$) of the probability that at least $t'$ (resp. $t'-1$) of the verifiers in $S \setminus \{i\}$ accept.
            \item If 
                \[ \widehat{p}_{t'}(r_i)-B(|S|-1, t', \eps) \geq \frac{(|S|-1)\nu}{k}, \] 
            then set $r^{*}_i = r_{\Query, i}$, $S \gets S \setminus \{i\}$ and return to the beginning of step 2. If 
                \[ \widehat{p}_{t'-1}(r_i)-B(|S|-1, t'-1, \eps) \geq \frac{(|S|-1)\nu}{k}, \]
            then set $r^{*}_i = r_{\Query, i}$, $S \gets S \setminus \{i\}$, $t' \gets t'-1$ and return to the beginning of step 2. 
        \end{enumerate}
        \item Output $S$, $t'$, and $\{r^{*}_{\Query, j}\}_{j \notin S}$.
    \end{enumerate}
    \item Output $S$, $t'$, and $\{r^{*}_{\Query, j}\}_{j \notin S}$.
\end{enumerate}
Clearly, step 2 terminates in at most $k$ iterations and in each iteration, we run at most $M_1 \cdot M_2$ executions of $\calB$ with $\widehat{V}^{k, t}$; thus, $\calT$ runs in time $\poly(\secp, k, M_1, M_2) = \poly(\secp, k, \nu^{-1})$.

Let $\gamma = \nu/10k$. Since each run in step 2(b) is independent, by Hoeffding's inequality (see Proposition \ref{proposition:hoeffding}) and a union bound over the $2kM_1$ estimates that are made, we have that the probability that $\widehat{p}_{t'}(r_i)$ or $\widehat{p}_{t'-1}(r_i)$ is ever off by more than $\gamma/2$ is at most
    \[ 4kM_1\exp(-M_2\gamma^2/2) \leq 2^{-\secp-1}. \]
Call the event that all estimates are accurate up to $\gamma/2$ error $\Good$; we have just shown that $\Pr[\Good] \geq 1-2^{-\secp-1}$.

It is straightforward to see that conditioned on $\Good$, the output of $\calT$ always satisfies
    \[ \Pr[\langle \calA, \widehat{V}^{k', t'}\rangle (1^{\secp}) = 1] \geq B(k', t', \eps) + \frac{k'}{k} \cdot \nu - \gamma = B(k', t', \eps) + \frac{10k'-1}{10k} \cdot \nu. \]

Now at each of the $k$ possible iterations of step 2, consider the set $D_i$ of coins $r_i$ such that $p_{t'}(r_i) > B(|S|-1, t', \eps) + \frac{(|S|-1)\nu}{k}+\gamma$ or $p_{t'-1}(r_i) > B(|S|-1, t'-1, \eps) + \frac{(|S|-1)\nu}{k}+\gamma$. Conditioned on $\Good$, the only way that $S$ is not updated is if all $M_1$ runs in step 2(b) miss $D_i$. If $\Pr[r_i \in D_i] \geq \gamma$, then this occurs with probability at most
    \[ (1-\gamma)^{M_1} \leq \exp(-\gamma M_1) \leq \frac{2^{-\secp}}{2k}. \]
Taking a union bound over at most $k$ iterations of step 2 implies that the probability of such an event ever occurring is at most $2^{-\secp-1}$. Thus, with probability at most $2^{-\secp-1}+2^{-\secp-1} = 2^{-\secp}$, no bad events occur and all the guarantees of Lemma \ref{lemma:corr_reduction} hold, as desired.
\end{proof}

\subsection{Proof of Lemma \ref{lemma:poly_tight_final_prob}}
\begin{proof}[Proof of Lemma \ref{lemma:poly_tight_final_prob}]
Let $\gamma := \nu/10k$. Let $(A_1, \ldots, A_k)$ denote the random variables corresponding to the verdicts of each subverifier in a random interaction between $\calB$ and $\widehat{V}^{k', t'}$. We define the quantities $Y := \sum_{j \neq 1} A_j$ and 
    \[ \eta(r) := \Pr_{\tau \gets \calH_m}[A_1 = 1 \mid r_{\Query, 1} = r]. \]
Note that $\calH_m$ accepts on iteration $\iter$ precisely when the first $\iter-1$ iterations all have $Y \neq t'-1$ and on iteration $\iter$, $Y = t'-1$ and $A_1 = 1$. Thus,
\begin{align*}
    \eta(r) &= \Pr[A_1 = 1 \land Y = t'-1 \mid r_{\Query, 1} = r] \sum_{\iter = 1}^{N} \bigl(1 - \Pr[Y = t'-1 \mid r_{\Query, 1} = r]\bigr)^{\iter-1} \\
    &= \Pr[A_1 = 1 \land Y = t'-1 \mid r_{\Query, 1} = r] \cdot \frac{1-\left(1-\Pr[Y = t'-1 \mid r_{\Query, 1} = r]\right)^N}{\Pr[Y = t'-1 \mid r_{\Query, 1} = r]}.
\end{align*}
Rearranging this identity gives
\begin{align*}
    &\Pr[A_1 = 1 \land Y = t'-1 \mid r_{\Query, 1} = r] -\eta(r)\Pr[Y = t'-1 \mid r_{\Query, 1} = r] \\
    =\,\, &\Pr[A_1 = 1 \land Y = t'-1 \mid r_{\Query, 1} = r] \cdot \bigl(1-\Pr[Y = t'-1 \mid r_{\Query, 1} = r]\bigr)^N.
\end{align*}
We now observe that $0 \leq \Pr[A_1 = 1 \land Y = t'-1 \mid r_{\Query, 1} = r] \leq \Pr[Y = t'-1 \mid r_{\Query, 1} = r]$, and that if $\Pr[A_1 = 1 \land Y = t'-1 \mid r_{\Query, 1} = r] \geq \gamma$, then 
    \[ \bigl(1-\Pr[Y = t'-1 \mid r_{\Query, 1} = r]\bigr)^N \leq (1-\gamma)^N \leq e^{-N\gamma} \leq e^{-2\secp} < \gamma. \]
Thus,
    \[ \Pr[A_1 = 1 \land Y = t'-1 \mid r_{\Query, 1} = r] -\eta(r) \Pr[Y = t'-1 \mid r_{\Query, 1} = r] < \gamma. \]
Putting this together, we see that
\begin{align*}
    &\Pr[A_1 + Y \geq t' \mid r_{\Query, 1} = r] \\
    = \,\, &\Pr[Y \geq t' \mid r_{\Query, 1} = r] + \Pr[A_1 = 1 \land Y = t'-1 \mid r_{\Query, 1} = r] \\
    \leq \,\, &\Pr[Y \geq t' \mid r_{\Query, 1} = r] + \eta(r) \cdot \Pr[Y = t'-1 \mid r_{\Query, 1} = r] + \gamma \\
    = \,\, &(1-\eta(r)) \cdot \Pr[Y \geq t' \mid r_{\Query, 1} = r] + \eta(r) \cdot \Pr[Y \geq t'-1 \mid r_{\Query, 1} = r] + \gamma.
\end{align*}
Let us denote by $\Good$ the set of coins $r$ such that all of the inequalities in Theorem \ref{thm:poly_round_tight} hold. Note that when $r \in \Good$, we have that
\begin{align*}
    &\Pr[A_1 + Y \geq t' \mid r_{\Query, 1} = r \in \Good] \\
    \leq \,\, &(1-\eta(r)) \cdot B(k'-1, t', \eps) + \eta(r) \cdot B(k'-1, t'-1, \eps) + (10k'-9)\gamma + \gamma.
\end{align*}
Since $\Pr[A_1 + Y \geq t' \mid r_{\Query, 1} = r] \leq 1$ for all $r$, we see that
\begin{align*}
    &\Pr[A_1 + Y \geq t' \mid r_{\Query, 1} = r] \\
    \leq \,\, &(1-\eta(r)) \cdot B(k'-1, t', \eps) + \eta(r) \cdot B(k'-1, t'-1, \eps) + (10k'-8)\gamma + \mathds{1}_{r \notin \Good}.
\end{align*}

Averaging over the choice of $r_{\Query, 1}$, we see that
\begin{align*}
    \Pr[A_1 + Y \geq t'] &= \bbE_r\left[\Pr[A_1 + Y \geq t' \mid r_{\Query, 1} = r]\right] \\
    &\leq B(k'-1, t', \eps) + \left(B(k'-1, t'-1, \eps) - B(k'-1, t', \eps)\right) \cdot \bbE_r[\eta(r)] \\
    &+ (10k'-8)\gamma + \Pr_r[r \notin \Good] \\
    &\leq \left(B(k'-1, t'-1, \eps) - B(k'-1, t', \eps)\right) \cdot \Pr_{\tau \gets \calH_m}[\Accept_1(\tau) = 1] \\
    &+ B(k'-1, t', \eps) + (10k'-7)\gamma.
\end{align*}
On the other hand, the hypothesis of Theorem \ref{thm:poly_round_tight} implies that
\begin{align*}
    \Pr[A_1 + Y \geq t'] &\geq B(k', t', \eps) + (10k'-1)\gamma \\
    &= (1-\eps)B(k'-1, t', \eps) + \eps B(k'-1, t'-1, \eps) + (10k'-1)\gamma \\
    &= B(k'-1, t', \eps) + \eps (B(k'-1, t'-1, \eps) - B(k'-1, t', \eps)) + (10k'-1)\gamma.
\end{align*}
Combining the two inequalities gives
\begin{align*}
    \left(B(k'-1, t'-1, \eps) - B(k'-1, t', \eps)\right) \cdot \left(\Pr_{\tau \gets \calH_m}[\Accept_1(\tau) = 1]-\eps\right) \geq 6\gamma.
\end{align*}
Observing that 
    \[ 0 < B(k'-1, t'-1, \eps) - B(k'-1, t', \eps) = \Pr_{X_1, \ldots, X_{k'-1} \sim \Ber(\eps)}\left[\sum_{j = 1}^{k'-1} X_j = t'-1\right] \leq 1, \]
we conclude that
    \[ \Pr_{\tau \gets \calH_m}[\Accept_1(\tau) = 1] \geq \eps + 6\gamma = \eps+\frac{3\nu}{5k}, \]
as desired.
\end{proof}

\end{document}